\documentclass[11pt,letterpaper]{article}
\usepackage[letterpaper,margin=1in]{geometry}
\usepackage[UKenglish]{babel}
\usepackage{amsmath,amssymb}
\usepackage{mathtools}
\usepackage{amsthm,thmtools,thm-restate}
\makeatletter
\@ifundefined{newcounteralias}{}{%
  \renewcommand\thmt@autorefsetup{%
    \@xa\def\csname\thmt@envname autorefname\@xa\endcsname\@xa{\thmt@thmname}%
  }%
}
\makeatother
\usepackage{graphicx}
\usepackage{xcolor}
\usepackage{caption}
\usepackage{paralist}
\usepackage{enumitem}
\usepackage[Algorithm]{algorithm}
\usepackage[noend]{algpseudocode}
\usepackage{cite}
\usepackage{xspace}
\usepackage{hyperref}
\usepackage[nameinlink,noabbrev,capitalize]{cleveref}

\algrenewcommand\algorithmicrequire{\textbf{Input:}}
\algrenewcommand\algorithmicensure{\textbf{Output:}}

\graphicspath{{./figures/}}

\title{A Sublinear Approximation Algorithm for Minimum Dilation Trees in the Plane}

\author{Sarita de Berg\thanks{IT University of Copenhagen, Denmark. \texttt{debe@itu.dk}}
\and
Jacobus Conradi\thanks{University of Copenhagen, Denmark. \texttt{jaco@di.ku.dk}}
\and
Peter Kramer\thanks{TU Braunschweig, Germany. \texttt{kramer@ibr.cs.tu-bs.de}}
\and
André Nusser\thanks{Université Côte d'Azur, CNRS, Inria, France. \texttt{andre.nusser@cnrs.fr}}
\and
Sampson Wong\thanks{University of Copenhagen, Denmark. \texttt{sampson.wong123@gmail.com}}
}

\date{}

\newcommand{\mysubpara}[1]{\par\vspace{0.3\baselineskip}\noindent{\textbf{#1}}\hspace{0.8em}}

\newtheorem{openproblem}{Open Problem}

\newtheorem{theorem}{Theorem}
\newtheorem{lemma}[theorem]{Lemma}
\newtheorem{observation}[theorem]{Observation}
\newtheorem{definition}[theorem]{Definition}
\newtheorem{corollary}[theorem]{Corollary}

\newtheorem{claim}[theorem]{Claim}
\AddToHook{env/lemma/begin}{\crefalias{theorem}{lemma}}
\AddToHook{env/claim/begin}{\crefalias{theorem}{claim}}
\AddToHook{env/remark/begin}{\crefalias{theorem}{remark}}
\AddToHook{env/observation/begin}{\crefalias{theorem}{observation}}
\AddToHook{env/corollary/begin}{\crefalias{theorem}{corollary}}

\crefname{observation}{Observation}{Observations}
\crefname{claim}{Claim}{Claims}
\crefname{openproblem}{Open Problem}{Open Problems}

\newenvironment{claimproof}
{\begingroup\begin{proof}[Proof of claim.]
}{
    \end{proof}\endgroup
}

\newcommand{\dist}{\ensuremath{\mathrm{d}}}
\newcommand{\dil}{\ensuremath{\mathrm{dil}}}
\newcommand{\supp}[1]{\ensuremath{\mathcal{P}_{#1}}}
\newcommand{\opt}{\ensuremath{\textsc{Opt}}\xspace{}}
\newcommand{\eps}{\ensuremath{\varepsilon}\xspace{}}
\def\len{\operatorname{len}}
\def\int{\operatorname{int}}

\newcommand{\F}{\ensuremath{\mathcal{F}}}
\newcommand{\T}{\ensuremath{\mathcal{T}}}
\newcommand{\J}{\ensuremath{\mathcal{J}}}
\newcommand{\junc}{\ensuremath{\mathrm{junc}}}

\newcommand{\spine}{\ensuremath{\Sigma}}
\newcommand{\spineset}{\ensuremath{\mathfrak{S}}}
\newcommand{\spines}[1]{\ensuremath{\mathrm{spines}(#1)}}
\newcommand{\buffer}[2]{\ensuremath{\mathrm{B}_{#1}(#2)}}

\begin{document}
    \maketitle
    \thispagestyle{empty}
    \begin{abstract}
        The \emph{dilation} of a geometric graph measures how much longer the path between pairs of points becomes when restricted to graph edges, rather than following the direct path through the ambient space.
        The \emph{minimum dilation tree} of a point set is the spanning tree with minimum dilation, where edge lengths in the tree are given by distances in the ambient space.
        In the Euclidean plane, computing the minimum dilation tree is NP-hard, but no hardness of approximation result is known. On the other hand, the minimum spanning tree is an $(n-1)$-approximation to the minimum dilation tree, but no asymptotically-better approximation algorithm is known for general point sets in the Euclidean plane.

        We give the first sublinear approximation algorithm for the minimum dilation tree in the Euclidean plane. Our approximation ratio is $\tilde{O}(n^{14/15})$ and our algorithm runs in polynomial time. This resolves an open problem proposed by Eppstein in 1996.
    \end{abstract}
    \vfill
    \begin{figure}[!h]
        \centering%
        \includegraphics{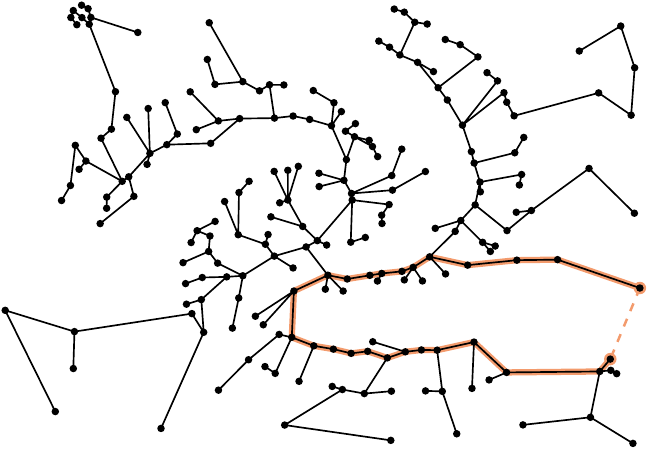}
        \caption*{Figure: A tree with near-optimal dilation ($\approx 9.72$), and the pair of points realizing this dilation.}
        \label{fig:front}
    \end{figure}

    \setcounter{page}{0}
    \newpage

    \section{Introduction}
    \label{sec:introduction}

    Graphs offer a compact, combinatorial way to represent distances of a point set, and the most compact such representation is a minimally connected graph, namely, a tree.
    Consequently, a fundamental question in theoretical computer science is how well a tree can represent distances of a point set.
    Dilation is by far the most common quality measure of how well distances are preserved by a graph, with smaller dilation implying a better representation.
    This work addresses the classical problem of computing minimum dilation trees.
    Formally, let~$S$ be $n$ points in the Euclidean plane, and let $\dist_T(\cdot,\cdot)$ denote the shortest path distance in a spanning tree~$T$ on~$S$.
    The dilation of $T$ is $\max_{u,v \in S} \frac {\dist_T(u,v)} {\|u-v\|}$.
    A minimum dilation tree of $S$ is a spanning tree on $S$ of minimum~dilation.

    Minimum dilation trees have attracted great interest across theoretical computer science. Depending on the community, they are known as tree spanners~\cite{DBLP:journals/algorithmica/BiloCGLP20,DBLP:journals/tcs/BrandstadtDLL04,DBLP:journals/algorithmica/BrandstadtDLLU07,DBLP:journals/siamdm/CaiC95,DBLP:conf/cocoa/CoutoC18,DBLP:journals/ipl/CoutoCJS22,DBLP:journals/jgaa/DasGW10,DBLP:journals/algorithmica/DraganK14,DBLP:journals/dam/FeketeK01,DBLP:journals/ipl/FominGL11,DBLP:conf/latin/GomezMW22,DBLP:journals/algorithmica/IwamaLO08,DBLP:journals/networks/LeL99,DBLP:journals/dam/LiebchenW08,DBLP:journals/ipl/MadanlalVR96},
    minimum dilation spanning trees~\cite{DBLP:journals/comgeo/AronovBCGHSV08,DBLP:conf/cocoon/BrandtGRS15,DBLP:journals/comgeo/CheongHL08,DBLP:books/el/00/Eppstein00},
    or minimum stretch spanning trees~\cite{DBLP:journals/siamcomp/EmekP08,DBLP:journals/amc/LinL20,DBLP:conf/mfcs/Peleg02,DBLP:conf/icalp/PelegR99,peleg2001low}. Minimum dilation trees are closely connected to other geometric constructions, such as spanners and tree covers; see Section~\ref{section:other_related_work}. Spanning trees with bounded dilation have also found applications in, for example, distributed computing~\cite{DBLP:conf/wdag/DemmerH98,DBLP:conf/wdag/GhodselahiK17,DBLP:conf/podc/HerlihyTW01,DBLP:conf/mfcs/Peleg02,DBLP:journals/dc/SharmaB14}.
    In 1996, Eppstein proposed the following open problems in his survey~\cite{eppstein1996spanning}, and later in the Handbook on Computational Geometry~\cite{DBLP:books/el/00/Eppstein00}.

    \begin{openproblem}
        \label{open_problem_1}
        Given a set of points in the Euclidean plane, is it possible to construct the minimum dilation tree, or an approximation to it, in polynomial time? Can a minimum dilation tree have edge crossings? How well is it approximated by the minimum spanning tree?
    \end{openproblem}

    By 2007, almost all of Eppstein's problems~\cite{eppstein1996spanning,DBLP:books/el/00/Eppstein00} on minimum dilation trees had been resolved. Eppstein~\cite{eppstein1996spanning} observed that the minimum spanning tree (MST) is an $O(n)$-approximation for the minimum dilation tree. Klein and Kutz~\cite{DBLP:conf/gd/KleinK06} constructed an example with seven points such that the minimum dilation tree has an edge crossing. Cheong, Haverkort and Lee~\cite{DBLP:journals/comgeo/CheongHL08} showed that the MST is an $\Omega(n)$-approximation, that there is a minimum dilation tree on five points with an edge crossing, and that constructing the minimum dilation tree is NP-hard.

    However, one problem remains elusive: approximating the minimum dilation tree. The following reformulation of Eppstein's problem has been stated numerous times~\cite{DBLP:journals/comgeo/AronovBCGHSV08,DBLP:journals/tcs/BiloGP12,DBLP:conf/esa/BuchinBGW24,DBLP:conf/compgeom/BuchinRS25,DBLP:journals/comgeo/CheongHL08} and is considered to be a major open problem in the field:

    \begin{openproblem}
        \label{open_problem_2}
        Is there an $o(n)$-approximation for the minimum dilation tree of a point set in the Euclidean plane?
    \end{openproblem}

    \cref{open_problem_2} remains wide open: the best lower bound~\cite{DBLP:journals/comgeo/CheongHL08} does not rule out a PTAS, and the best upper bound is still the simple $O(n)$-approximation given by the minimum spanning tree~\cite{eppstein1996spanning}. The difficulty of approximating the minimum dilation tree is a significant obstacle in several other problems, for example, sparsest spanners~\cite{DBLP:journals/comgeo/AronovBCGHSV08}, minimum dilation graph augmentation~\cite{DBLP:conf/esa/BuchinBGW24}, and bounded treewidth spanners~\cite{DBLP:conf/compgeom/BuchinRS25}.

    \subsection{Our contribution}
    \label{section:contributions}

    We give
    the first sublinear approximation for the minimum dilation tree, resolving \cref{open_problem_2}.

    \begin{restatable}[Main theorem]{theorem}{main} \label{thm:main}
        There is an $O(n^{14/15}\log^3 n)$-approximation algorithm for the minimum dilation tree problem with running time in $O(n^5\log n)$.
    \end{restatable}

    In order to achieve this, we introduce several new insights on the structure of minimum dilation trees in the plane.
    In particular, we categorize different point set structures that inherently lead to a dilation of $\Omega(n^c)$ of the minimum dilation tree,  for some constant $c>0$.
    On the other hand, we show that an absence of these structures enables the construction of a minimum dilation tree with sublinear dilation.
    We believe these insights to be of independent interest.

    \subsection{Other related work}
    \label{section:other_related_work}

    There is a plethora of related work on similar but more relaxed settings like bounded spread point sets, denser-than-tree spanners, or using more than one tree to represent the embedding.

    \mysubpara{Bounded spread.} Bădoiu, Indyk and Sidiropoulos~\cite{DBLP:conf/soda/BadoiuIS07} study computing the minimum dilation tree of point sets from an arbitrary metric space, restricted to the special case of point sets with bounded spread~$\Delta$, which
    is the ratio between its maximum and minimum interpoint distance.
    They provide an $\alpha (\opt\cdot \log n)^{O(\log_{\alpha} \Delta)}\,$-approximation algorithm for any~${\alpha \geq 1}$, where $\opt$ is the dilation of the minimum dilation tree.
    We note that the exponent in the approximation ratio of~\cite{DBLP:conf/soda/BadoiuIS07} is $c_1 \log_{\alpha} \Delta$ with $c_1 \geq 7$.
    If $\Delta\le n^{c_2}$ for a
    constant $c_2$, then choosing $\alpha=\sqrt n$ yields a tree of dilation
    $O(\sqrt n(\opt\log n)^{2c_1c_2})$.
    Taking the better of this tree and the minimum spanning
    tree gives an
    $O(n^{1-1/(4c_1c_2)}\log n)$-approximation.
    Hence, their approximation guarantee converges towards a linear approximation with increasing spread~$\Delta$.
    Nayyeri and Raichel~\cite{Nayyeri_tree_embeddings} present a parameterized $(1+\varepsilon)$-approximation algorithm. For points in the Euclidean plane, their result gives a PTAS when the optimal dilation is constant and the spread is polynomially bounded. However, their running time depends exponentially on $\opt$ and polynomially on the spread, and thus does not yield a polynomial-time approximation algorithm for arbitrary point sets.

    \mysubpara{Sparse spanners.} In $\mathbb R^2$, minimum dilation trees are very sparse geometric spanners, which relax the restriction of $T$ being a tree to $T$ being any graph on $S$. Geometric spanners are well studied, with many constructions known for $(1+\varepsilon)$-dilation spanners with $O(n/\varepsilon^2)$ edges. See the textbook~\cite{DBLP:books/daglib/NarasimhanSmid} for an overview. Aronov, de Berg, Cheong, Gudmundsson, Haverkort, Smid and Vigneron~\cite{DBLP:journals/comgeo/AronovBCGHSV08} present an $O(n/(\ell+1))$-dilation spanner with $(n-1+\ell)$ edges for all $\ell \geq 1$. This result can be seen as a smooth transition from minimum dilation trees to $(1+\varepsilon)$-spanners. Augustine, Eppstein, and Wortman~\cite{Eppstein_min_dilation_star,Augustine_min_dilation_star} study the even more restricted setting of minimum dilation stars.

    \mysubpara{Tree covers.} A Euclidean tree cover is a collection of trees that together approximate an~\mbox{$n$-point} Euclidean metric space. In particular, for every pair of points there exists a tree in the $(1+\varepsilon)$-tree cover such that their distance in the tree is a $(1+\varepsilon)$-approximation of their Euclidean distance. The celebrated Dumbbell theorem of Arya, Das, Mount, Salowe and Smid~\cite{DBLP:conf/stoc/AryaDMSS95} states that any $n$ points in the Euclidean plane admit a tree cover with $O(\varepsilon^{-2} \cdot \log (1/\varepsilon))$ trees and~$(1+\varepsilon)$-dilation. Chang, Conroy, Le, Milenković, Solomon, and Than~\cite{DBLP:journals/dcg/ChangCLMST26} improve the number of trees in the tree cover to $O(\varepsilon^{-1} \cdot \log(1/\varepsilon))$. Chan's shifting lemma~\cite{DBLP:journals/dcg/Chan98} implies a Euclidean tree cover with 3 trees and $6\sqrt 2$ dilation, as observed by~\cite{DBLP:conf/soda/0001MS026}. Le, Milenković, Solomon and Zhang~\cite{DBLP:conf/soda/0001MS026} and Bikeev, Kupavskii, and Turevskii~\cite{DBLP:conf/soda/BikeevKT26} independently proved that there exists a Euclidean tree cover with 2 trees and $O(1)$ dilation. Minimum dilation trees are tree covers of size 1.

    \mysubpara{Unweighted graphs.} Minimum dilation trees have been studied extensively in unweighted graph metrics, e.g., in the shortest path metric on graphs with unit edge weights.
    Cai and Corneil~\cite{DBLP:journals/siamdm/CaiC95} consider the problem of deciding whether an unweighted graph admits a spanning tree with dilation~$t$. They provided a linear time algorithm for $t \leq 2$, and showed that the problem is NP-complete for $t \geq 4$. For $t=3$, the problem remains open. Special cases of the tree 3-spanner problem have been considered, for example, in split graphs and complements of bipartite graphs~\cite{DBLP:phd/ca/Cai92}, interval, permutation and regular bipartite graphs~\cite{DBLP:journals/ipl/MadanlalVR96}, planar graphs~\cite{DBLP:journals/dam/FeketeK01}, directed path graphs~\cite{DBLP:journals/networks/LeL99}, very strongly chordal graphs~\cite{DBLP:journals/tcs/BrandstadtDLL04}
    and bounded degree graphs~\cite{DBLP:journals/ipl/FominGL11}. The problem of approximating the minimum dilation tree has also been studied. There is an $O(\log n)$-approximation algorithm~\cite{DBLP:journals/siamcomp/EmekP08}, and an NP-hardness result for a $(2-\varepsilon)$-approximation algorithm~\cite{DBLP:conf/isaac/Galbiati01,DBLP:journals/dam/LiebchenW08} for any $\varepsilon > 0$.

    \begin{figure}
        \centering
        \includegraphics{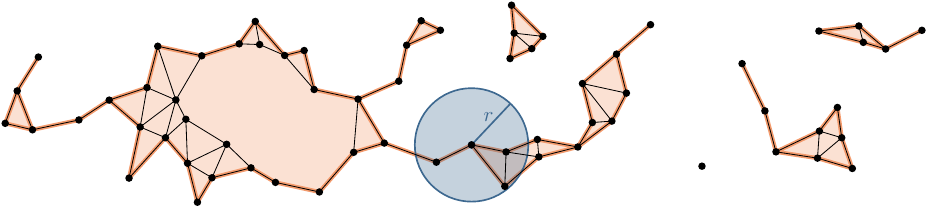}
        \caption{The Delaunay triangulation after removing all edges of length greater than $r$, and the corresponding $r$-support of the point set $S$ consisting of weakly simple polygons and a singleton.}
        \label{fig:r-support}
    \end{figure}

    \section{Preliminaries}
    \label{sec:preliminaries}

    Let $S$ be a set of $n$ points in the plane and let~${G=(S,E)}$ be an undirected geometric graph whose vertices are the points in~$S$, where the length of an edge between points $u$ and $v$\ is the Euclidean distance between its endpoints~$\|{u-v}\|$.
    For any two points~${p,q\in S}$, we denote by~$\dist_G(p,q)$ the length of a shortest $pq$-path $\pi_G(p,q)$ in~$G$, or~${\infty}$ if no such path exists.

    \mysubpara{Dilation.}
    For a pair of distinct points $p$ and $q$ in $S$, the \emph{dilation} of $p$ and $q$ with respect to $G$ is the ratio $\dil_G(p,q)= \frac{\dist_G(p,q)}{\lVert p-q\rVert}$.
    The dilation $\dil(G)$ of the graph $G$ is defined as the maximum dilation $\max_{p,q\in S}\dil_G(p,q)$ of any pair of points in $S$.
    A \emph{minimum dilation tree} $T$ of~$S$ is a geometric graph on $S$ that is a tree that minimizes the dilation.
    We denote by $\opt$ the dilation of such a minimum dilation tree, i.e., $\opt = \min \{\dil(T) : \text{$T$ spanning tree of $S$}\}$.

    \mysubpara{Polygons.} We will consider weakly simple polygons containing a subset of the points in~$S$.
    A polygon $P$ is given by a closed boundary walk $\partial P$ of line segments; it is \emph{weakly simple} if (i) for any~${\eps>0}$ its vertices can be $\eps$-perturbed to obtain a simple polygon and (ii) its complement is path-connected\footnote{This is a slightly more restrictive subclass of what are typically called weakly simple polygons in literature~\cite{DBLP:journals/dcg/AkitayaAET17}, as it excludes self-touching weakly simple polygons whose complement has a bounded component.}.
    Repeated visits of the walk to a point count as distinct vertices, so $P$ may have degenerate 1-dimensional parts (see~\Cref{fig:r-support}). We interpret $\partial P$ and $\int(P)=P\setminus\partial P$ with respect to the $\eps$-perturbed polygon.
    For a polygon $P$, we denote by~$\dist_P(p,q)$ the geodesic distance between $p$ and $q$, which is the length of the shortest path~$\pi_P(p,q)$ contained within the polygon~$P$.

    We denote by $D(p,r)$ the Euclidean disk of radius $r$ centred at the point $p$. We similarly denote for a weakly simple polygon $P$ and a point $p\in P$ the \emph{geodesic} disk of radius $r$ by~$D_P(p,r)$.

    \mysubpara{Support.} We define the $r$-support to capture the shape of a point set for any~${r\geq0}$, see~\Cref{fig:r-support}.

    \begin{definition}[$r$-support]
    For any $r\geq 0$, the \emph{$r$-support}~$\supp{r}(S)$ of a point set $S$ is the set of inclusion-wise maximal,
    weakly simple polygons bounded by the outer boundary walks of the connected components of the Delaunay triangulation after removing all edges of length greater than~$r$.
    \end{definition}

    With slight abuse of notation, we write $\supp{r}$ instead of $\supp{r}(S)$ whenever $S$ is clear from context, and we may express Boolean operations over the union of polygons as, e.g., $X\cap \supp{r}=X\cap\left(\bigcup_{P\in\supp{r}}P\right)$.
    Note that the $r$-support may contain degenerate polygons, which are defined by just one point.

    \mysubpara{Thinness.}
    Our construction of low dilation trees is guided by the geometric shape of the $r$-support. For this we quantify the ``tree-likeness'' of a weakly simple polygon.
    \begin{definition}[$\Lambda$-thin polygon and point set]
    A weakly simple polygon $P$ is \emph{$\Lambda$-thin} if the radius of the largest disk contained in $P$ is at most $\Lambda$.
    A point set $S$ is \emph{$\Lambda$-thin} if for all $r > 0$, every weakly simple polygon $P$ in the $r$-support $\supp{r}$ is $\Lambda r$-thin.
    \label{def:thinness}
    \end{definition}

    \section{Technical overview}
    \label{sec:tech_overview}

    We first consider several examples. This is to build towards our approach, and to highlight some peculiarities of minimum dilation trees.

    \begin{figure}
        \centering
        \includegraphics{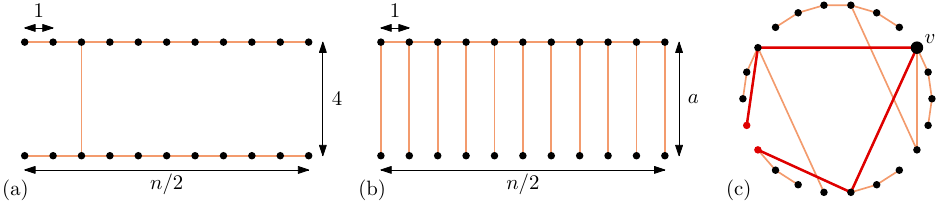}
        \caption{(a) The MST of a set of points lying on two horizontal lines at vertical distance $4$. (b) A comb with dilation $\Theta(a)$. (c) For points on a circle there is no tree with dilation $o(n)$.
        }
        \label{fig:examples}
    \end{figure}

    \begin{description}
        \item[MST upper bound \cite{eppstein1996spanning}:] The MST has dilation at most $n-1$. To see this, consider a pair of points $p,q$ at distance $d$ that are \emph{not} connected by an edge of the MST. Let $d$ be the Euclidean distance between $p$ and $q$. The shortest path between $p$ and $q$ in the MST consists of at most $n-1$ edges of length at most $d$, as otherwise the edge $\overline{pq}$ would have been in the MST. Hence the dilation of the MST is at most $\frac{(n-1)d}{d}=n-1$.

        \item[MST lower bound \cite{DBLP:journals/comgeo/CheongHL08}:]  Consider the set of $n$ points $S=\{1,\ldots,n/2\}\times \{0,4\}$, see \Cref{fig:examples}(a). The MST connects the horizontal pairs at distance $1$ and picks a vertical connection of length~$4$. So, either the left-most or right-most vertical pair has distance at least $2(n/4)+4$ in the tree. The dilation is thus at least $\frac{2(n/4)+4}{4}\geq n/8$. However, connecting the points by a comb as in \Cref{fig:examples}(b), we obtain a dilation of $9$. Therefore, the MST is a $\Theta(n)$-approximation.

        \item[Comb lower bound:]
        Consider for any $1\leq a\leq n$ the set of $n$ points $S_a=\{1,\ldots,n/2\}\times \{0,a\}$.
        Consider two trees $T_1,T_2$ on the points $S_a$.
        $T_1$ connects the horizontal lines of points, and includes a single vertical connection, as in \Cref{fig:examples}(a). $T_2$ connects the points by a comb as in \Cref{fig:examples}(b). The dilation of $T_1$ is $\Theta(\frac{n}{a})$, while the dilation of $T_2$ is $\Theta(a)$. Therefore, we can obtain a tree with dilation $O(\min(n/a,a))$ by choosing either the ``MST'' or the ``comb'' construction. It turns out that this pair of constructions is optimal: any spanning tree of $S_a$ has dilation $\Omega(\min(n/a,a))$.

        \item[Circle lower bound \cite{eppstein1996spanning}:] Consider a set $S$ of $n$ points evenly spread on a unit circle. The dilation of any tree on $S$ is $\Omega(n)$: There is a vertex $v$ whose removal splits the tree into components of size at most $2n/3$, which gives two points in different components that are adjacent in the cycle and far from $v$, see \Cref{fig:examples}(c). The path in the tree connecting these two points has to travel to $v$ and back, and thus their dilation is in $\Omega(n)$. The points do not necessarily have to be spread equally along the cycle for the lower bound to hold, as long as the distance between adjacent points is $O(1/n)$.
    \end{description}

    We are now ready to give an overview of our approximation algorithm and the tools involved.
    Our algorithmic approach relies heavily on the identification of lower bounds for $\opt$.
    The circle lower bound and comb lower bound are examples that act as obstructions preventing a point set from admitting a spanning tree with small dilation.
    We generalize these lower bound obstructions in order to capture all structures that prevent a good minimum dilation tree.
    This helps us in the following way: if we are able to find some structure in the point set $S$ that certifies that $\opt\geq n^{c}$, for some small constant $c$, then the minimum spanning tree of $S$, having a dilation of $n-1$~\cite{eppstein1996spanning}, has a dilation of at most $n^{1-c} \cdot \opt$.
    We denote this $n^c$ by the parameter $\Lambda$, for which we will only set the exact parameter $c$ at the very end of the paper.
    Throughout this paper, we will assume~$\Lambda \geq 6$.
    Note that these structures can even appear as subsets in the point set, since discarding irrelevant points can only increase the dilation by at most a constant factor~\cite{DBLP:conf/soda/Gupta01}.

    There are two types of lower bounds, or obstructions, used by our algorithm. The first obstruction is a generalization of the circle lower bound. It gives us a strong structure on the point set and, in particular, on the $r$-support for any $r$. We will show that if the point set $S$ is not $\Lambda$-thin, then $\opt\in \Omega(\Lambda)$ (\Cref{cor:circle-test}).
    If the point set is $\Lambda$-thin instead, then we are able to compute a tree structure of each polygon in $\supp{r}$ for any value $r$. This tree structure consists of a partition into polygons of geodesic diameter $\Lambda r$, whose intersection graph is a tree. Computing this for every value\footnote{For point sets of super-quadratic spread, we compress this sequence, so that at most $n^2$ values of $\Lambda^{k}$ remain.} of $r=r_0,\Lambda r_0, \Lambda^2 r_0, \ldots$, where $r_0$ is the smallest interpoint distance in~$S$, we obtain a sequence\footnote{This is very similar to the filtration induced by the $\alpha$-complex.}of supports $\supp{r_0},\supp{\Lambda r_0} \ldots$, where each polygon in $\supp{r}$ is nested in a polygon of $\supp{\Lambda r}$ and endowed with its own tree structure.

    The second obstruction is a generalization of the comb lower bound. In particular, we test whether the sequence of supports $\supp{r}$ is \emph{consistent} across increasing values of $r$, meaning that the forests have a similar structure.
    To this end, we test at scale~$r$ whether there is some tree-like piece of $\supp{r}$ that is spanned by at least two tree-like pieces of $\supp{r/\Lambda^2}$.
    Intuitively, this means that there is a structure similar to the comb lower bound (with $a\approx \Lambda$) present in the point set.
    If this is the case, we show that $\opt \in \Omega(\Lambda)$ (\Cref{lem:the-hard-obstruction}).

    Hence, our core algorithm only needs to handle consistent sequences of supports.
    Let us consider a first attempt at developing such an algorithm.
    Given a consistent sequence of supports, the algorithm considers the scales $r=r_0,\Lambda r_0, \Lambda^2 r_0, \ldots$ in sequence.
    At every scale $r$
    and for each connected component of the $r$-support, our goal is to build a tree that connects the trees of the previous scale without modification, and follows the forest structure of the $r$-support closely.
    The result of phase~$r$ is a forest that connects any pair of points~$(p,q)$ with Euclidean distance at most~$r$ such that the shortest path in this forest is in $D\!\left(p,{\Lambda^2\|p-q\|}\right)$, i.e., local. During phase~$r$ we only add edges of length at most $\Lambda^2r$ and any such edge connects two disjoint connected components of~$\supp{r/\Lambda}$, and thus, via Euclidean packing arguments and constant-depth recursion, we can bound the length of the shortest path in the forest, giving dilation $\Lambda^{O(1)}+n/\Lambda$ for $p$ and $q$. For any pair of points with Euclidean distance~$d$, there is a scale $r$ such that $d\in[r/\Lambda,r]$, which will make sure that, in the end, any pair of points has good dilation.

    \begin{figure}
        \centering
        \includegraphics{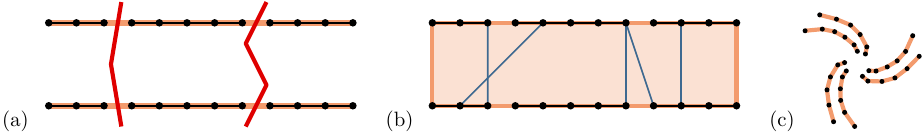}
        \caption{(a) The red cut splits the $r$-support (orange) such that (b) at the next scale we can include the blue edges to get a good tree. (c) A more complex point set with similar structures.}
        \label{fig:intuition}
    \end{figure}

    This attempt has the following problem:
    Aggressively connecting all pairs with distance at most~$r$, without considering the next scale~$\Lambda r$, might make it impossible for the next scale to connect close pairs locally.
    In fact, \Cref{fig:intuition}(a), which is similar to the MST lower bound, and \Cref{fig:intuition}(c), which shows a similar, but more complex instance which consists of three twisted copies of the same instance, show instances where this approach does not work.
    In \Cref{fig:intuition}(a), the $r$-support for $r=1$ consists of two horizontal segments, while any good spanning tree requires us to make many top to bottom connections.

    We can turn our failed attempt into a working algorithm as follows.
    We first look at~$\supp{\Lambda r}$, and compute a set of cuts of the nested polygons $\supp{r}$
    at regular intervals (w.r.t.\ the tree-structure of~$\supp{\Lambda r}$) that only separate points at distance at least~$r/\Lambda$.
    This cut tells scale $r$ that it should only worry about pairs of points within distance at most $r$ that come from the same cut-piece.
    Pairs from different pieces are claimed by the next scale, which will make sure that close pairs of points ignored at scale $r$ are both connected and \emph{locally} connected, as in \Cref{fig:intuition}(b).
    The structure, and choice of parameters prevents a propagation of influences by cuts, that is, the cut of~$\supp{r}$ induced by $\supp{\Lambda r}$ does not cut through any connected component of~$\supp{r/\Lambda}$.
    Conversely, a pair of points at distance $d\in[r/\Lambda,r]$ is always connected after scale~$\Lambda r$ has been handled.

    After the last scale has been handled, we obtain a tree, where each pair of points is connected locally (w.r.t.\ their distance), which, together with Euclidean packing arguments and constant-depth recursion, yields a tree with dilation at most $O((\Lambda^{O(1)}+n/\Lambda) \log^3 n)$. Choosing $\Lambda=n^{1/k}$ for a suitable $k$ gives a dilation of at most $O(n^{1-1/k} \log^3 n)$
    (\Cref{thm:main}). An interesting observation that is of independent interest is that the two obstructions in \Cref{cor:circle-test} and \Cref{lem:the-hard-obstruction} completely classify the sets of points for which $\opt$ is large: For sufficiently large $C$ and $k$, \emph{any} instance with $\opt > C n^{1-1/k} \log^3n$ \emph{must have} either a circle or a comb obstruction.

    \mysubpara{Organization.} In \Cref{sec:obstructions} we generalize the circle and comb lower bound. We first handle the easier of the two, the circle lower bound, in \Cref{sec:circle_obstructions}, before identifying the tree-shape of $r$-supports in \Cref{sec:treeifying}. Then, we formalize the generalized comb lower bound in \Cref{sec:combs}. In \Cref{sec:consistent} we define the sequence of nested forest-like polygons that is consistent. Finally, in \Cref{sec:algorithm}, we present our algorithm that uses all of the identified structures to explicitly construct a tree of sublinear dilation. We provide some concluding remarks in \Cref{sec:concluding_remarks}.
    Throughout the paper we defer some proofs to the appendix in order to focus on the most crucial parts.

    \section{Obstructions}\label{sec:obstructions}

    In this section, we generalize the two lower bounds briefly discussed in \Cref{sec:tech_overview}.
    We first discuss the (easier) circle lower bound. We relate it to the $\Lambda$-thinness of the point set and show how to algorithmically detect $\Lambda$-thinness by examining all distinct $r$-supports of a point set.
    If such a lower bound does not exist, we show that every $r$-support permits a tree-like structure which we formalize with the notion of \emph{junction trees}.
    This gives us a framework to formalize and generalize the comb lower bound.

    \subsection{Circle obstructions}\label{sec:circle_obstructions}

    We first generalize the circle lower bound to weakly simple polygons that contain a large disk.

    \begin{lemma}\label[lemma]{lem:circle-test}
        Consider a weakly simple polygon  $P=p_1,\ldots,p_m,p_1$ on a subset $\{p_1,\ldots,p_m\}\subseteq S$
        and let~${r = \max_i\|p_{i}-p_{i+1}\|}$, where indices are taken modulo $m$. Suppose that the disk $D(x,R)$ centred at $x$ with radius $R$ is contained within $P$. Then $\opt\ge 2R/r$.
    \end{lemma}

    \begin{proof}
        Let $T$ be any spanning tree on $S$, and consider its straight-line embedding. After an arbitrarily small perturbation of $x$, we may assume that $x\notin T$. Since $T$ is a tree, any embedding of a closed walk in $T$ is contractible in $\mathbb{R}^2\setminus\{x\}$.

        For $i \in [m]$, let $Q_i$ be the unique shortest path in $T$ from $p_i$ to $p_{i+1}$, see~\cref{fig:cycle_proof}.
        Suppose, for the sake of contradiction, that $\dist_T(p_i,p_{i+1})<2R$ for all $i$. Then ${Q_i\subset D(p_i,R)\cup D(p_{i+1},R)}$.
        Since~$D(x,R)$ lies inside $P$, the disks $D(p_i,R)$ and $D(p_{i+1},R)$ do not contain $x$.
        Hence $Q_i$ is homotopic to the segment $\overline{p_ip_{i+1}}$ in $\mathbb{R}^2\setminus\{x\}$, relative to endpoints.

        Concatenating these homotopies, $Q \coloneqq Q_1\oplus\cdots\oplus Q_m$ is homotopic in $\mathbb{R}^2\setminus\{x\}$ to the cycle~$\partial P$.
        However, the point $x$ is inside the plane cycle $\partial P$, and so $\partial P$ is not contractible in $\mathbb{R}^2\setminus\{x\}$, whereas~$Q$ is contractible in $\mathbb{R}^2\setminus\{x\}$, as it is a closed walk in $T$.
        Thus, by contradiction, we have that $\dist_T(p_i,p_{i+1})\ge 2R$ for some $i$. As $\|p_i-p_{i+1}\|\le r$, we obtain
        \[
            \dil(T)\ge \frac{\dist_T(p_i,p_{i+1})}{\|p_i-p_{i+1}\|}\ge \frac{2R}{r}.\qedhere
        \]
    \end{proof}

    \begin{figure}
        \centering
        \includegraphics{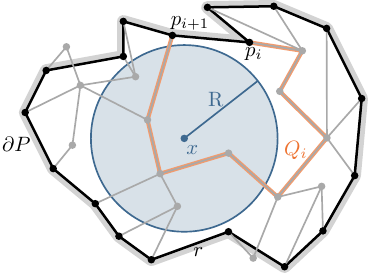}
        \caption{The boundary $\partial P$ connects the points $\{p_1,\ldots,p_m\}$. The path $Q_i$ (orange) between $p_i$ and $p_{i+1}$ in the tree $T$ (grey) goes ``around'' $x$ and thus is at least~$2R$ long.}
        \label{fig:cycle_proof}
    \end{figure}

    The lemma can directly be applied to the weakly simple polygons in the $r$-support to obtain:

    \begin{theorem}[Circle obstruction]\label{cor:circle-test}
        Let $r >0$. If there is a weakly simple polygon in the $r$-support~$\supp{r}$ that contains a disk of radius $R$, then $\opt \geq 2R/r$.
    \end{theorem}
    \begin{proof}
        Applying \Cref{lem:circle-test} to the weakly simple polygon in $\mathcal{P}_r$ that contains a disk of radius $R$ implies the theorem.
    \end{proof}

    Recall that if, for all $r$, the $r$-support of the point set $S$ does not contain a disk of radius~$\Lambda r$, then $S$ is said to be $\Lambda$-thin.
    We show that $\Lambda$-thinness can be verified in quadratic time.

    \begin{observation}[Number of distinct $r$-supports]
        \label{obs:number-of-distinct-supports}
        For a fixed set $S$ of $n$ points and all choices of~$r\geq 0$, there exist at most $3n$ distinct $r$-supports.
    \end{observation}
    \begin{proof}
        Each $r$-support considers a subset of edges of the Delaunay triangulation, which form a nested sequence of sets for increasing values of $r$. As the Delaunay triangulation contains at most $3n$ edges, there are at most $3n$ such distinct subsets used to construct an $r$-support.
    \end{proof}

    \begin{restatable}{lemma}{lemThinObstructionCheck}
    One can determine in $O(n^2)$ time whether a given set $S$ of $n$ points is $\Lambda$-thin. If~$S$ is not $\Lambda$-thin, then $\opt \geq 2\Lambda$.
    \label[lemma]{lem:thin-obstruction-check}
    \end{restatable}

    \begin{proof}[Proof sketch.]
    Recall that the number of distinct $r$-supports is bounded by the number of edges in the Delaunay triangulation~(\cref{obs:number-of-distinct-supports}).
    We compute the Delaunay triangulation and can then enumerate all distinct supports efficiently.
    For each value of $r$, we determine the radius of the largest inscribed disk in~$\supp{r}$ in linear time~\cite{DBLP:journals/dcg/ChinSW99}.
    By~\cref{cor:circle-test}, if this value exceeds~$\Lambda r$, the claimed lower bound of $(2\Lambda r)/r=2\Lambda$ on the achievable dilation holds.
    \end{proof}

    \subsection{\boldmath Converting a $\Lambda$-thin polygon into a tree}\label{sec:treeifying}

    From here on, we consider only $\Lambda$-thin point sets.
    Next, we construct the so called $\alpha$-tree partition. We will cluster the polygon into polygonal pieces of size~$\alpha$, such that the dual graph of the pieces is a tree. See \Cref{fig:tree-partition}(a) for an example.

    \begin{figure}%
        \centering%
        \includegraphics{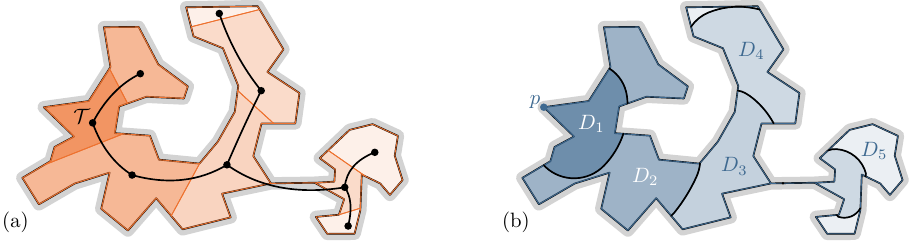}%
        \caption{(a) An $\alpha$-tree partition and the corresponding intersection graph. (b) The family of geodesic disks used in the construction of an $\alpha$-tree partition due to~\cref{thm:tree-partition}.}
        \label{fig:tree-partition}
    \end{figure}

    \begin{definition}[Tree partition]\label[definition]{def:tree-partition}
        Let $P$ be a weakly simple polygon. An \emph{$\alpha$-tree partition} of $P$ is a partition\footnote{
            A set of pieces $\{N_i\}$ with $N_i \subset P$ is a partition of $P$ if their union covers~$P$ and the vertices of the pieces $N_i$ can be $\varepsilon$-perturbed within $P$ for any $\varepsilon>0$ such that they become pairwise disjoint, weakly simple polygons.}
        of $P$ into pieces $N_i\subset P$, such that each piece $N_i$ is a weakly simple polygon with geodesic diameter at most $\alpha$, and the intersection graph of the pieces $N_i$ is a tree~$\mathcal{T}$.
    \end{definition}

    It will also be useful for our tree partition to have the so-called distance preserving property. This property relates the distance $\dist_P(\cdot,\cdot)$ in the polygon $P$ to the hop distance $\dist_{\mathcal T}(\cdot, \cdot)$ in the tree~$\mathcal{T}$. Specifically, it states that if two pieces are far in $\mathcal T$, they must also be far in~$P$.

    \begin{definition}[Distance preserving]
        \label{definition:distance_preserving}
        An $\alpha$-tree-partition is \emph{$c$-distance-preserving} if for any two points~${p \in N_i}$ and $q \in N_j$, it holds that
        \[
            c \cdot \dist_P(p,q) \geq \alpha\cdot(\dist_{\T}(N_i,N_j)-2).
        \]
    \end{definition}

    \begin{restatable}{lemma}{treePartition}\label{thm:tree-partition}
        Let $S$ be a $\Lambda$-thin point set with $\Lambda \geq 6$ and let $P\in\supp{r}(S)$ for $r>0$.
        In~$O(n\log n)$ time, we can compute an $18$-distance-preserving $(18\Lambda r)$-tree partition of $P$ into at most $2n-1~$~pieces.
    \end{restatable}
    \begin{proof}[Proof sketch.]
    The high-level idea for constructing a tree partition is to first compute a sequence of geodesic disks $D_1,D_2,\ldots$ centred at a vertex $p$ of~$P$ where the radius of the disk increases by~$\Theta(\Lambda r)$ in every step, see~\Cref{fig:tree-partition}(b).
    We then rectify the boundary of each disk $D_i$ by connecting the endpoints of each connected component of
    $\partial D_i\cap \int(P)$
    by a shortest path, as in \Cref{fig:tree-partition}(a).
    The final partition is obtained by cutting the polygon along the boundaries of these rectified disks.
    To prove the desired properties of the partition, we first show that for \emph{any} radius, the geodesic distance between the endpoints of a connected component of~${\partial D_i\cap \int(P)}$ is bounded by~$6\Lambda r$ using the thinness of $S$.
    Using this, we prove that the boundaries of successive disks~$D_i$ and $D_{i+1}$ are pairwise at distance~$\Theta(\Lambda r)$ even after rectification, and bound the diameter of each piece by~$O(\Lambda r)$.
    With these two bounds, we can then finally prove the distance-preserving property and bound the number of pieces.
    \end{proof}

    Next, we select a subset of pieces from our tree partition that act as checkpoints; we call these pieces \emph{junctions}.
    Removing the junctions from the polygon, we obtain weakly simple polygons that we call \emph{faces}.
    We call such a partition into junctions and faces a \emph{junction tree}, see \Cref{fig:junction-tree} for an example.
    The desirable properties of the junction tree are that the junctions and faces are not too big and that the junctions are not too close to one another.
    Next, we provide the formal definition.

    \begin{figure}
        \centering
        \includegraphics{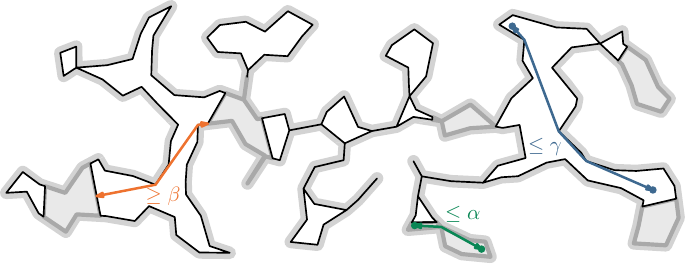}
        \caption{The junctions (grey) and faces (white) in a junction tree. The diameter of a junction is at most $\alpha$ (green), the diameter of a face is at most $\gamma$ (blue), and the distance between any two junctions is at least $\beta$ (orange).
        }
        \label{fig:junction-tree}
    \end{figure}

    \begin{definition}[Junction tree]
        \label[definition]{def:junction_tree}
        Let $P$ be a weakly simple polygon. An $(\alpha,\beta,\gamma)$-junction tree of~$P$ consists of two finite sets $\J$ (junctions) and $\F$ (faces) of weakly simple polygons such that
        \begin{itemize}[nolistsep]
        \item the polygons $\J\cup \F$ are a partition of~$P$ whose intersection graph is a tree,
        \item no two junctions in~$\J$ share a boundary; the same holds for faces~$\F$,
        \item the geodesic diameter of each polygon in $\J$ is at most $\alpha$,
        \item the geodesic distance between any two distinct polygons in $\J$ is at least $\beta$,
        \item the geodesic diameter of each polygon in $\F$ is at most $\gamma$.
        \end{itemize}
    \end{definition}

    Next, we construct the junction tree for polygons in the $r$-support.

    \begin{lemma}\label[lemma]{thm:thm-junction-tree}
        Let $S$ be a $\Lambda$-thin point set with $\Lambda\geq 6$.
        Let $P$ be a polygon of the \mbox{$r$-support}~$\supp{r}(S)$.
        In $O(n\log n)$ time we can compute a $(18\Lambda r,\frac{2}{3} \Lambda^2r,(54\Lambda^2\log n)r)$-junction tree of $P$.
    \end{lemma}

    \begin{proof}
        Consider an $18$-distance-preserving $(18\Lambda r)$-tree partition and a heavy-light decomposition~\cite{DBLP:conf/stoc/SleatorT81} of its intersection graph $\mathcal{T}$, rooted arbitrarily.
        In a heavy-light decomposition, each edge of a tree is labelled as either heavy or light: for each vertex, the edge to its child with the most descendants is marked as heavy, the edges to the other children are marked as light.
        We process the heavy paths from top to bottom, \emph{marking} every $\lceil \Lambda\rceil$'th vertex such that the first~${\lceil \Lambda\rceil-1}$ vertices of each heavy path are unmarked.
        Let $\J$ consist of the polygons corresponding to marked vertices, and let $\mathcal{F}$ contain, for each connected component of unmarked vertices, the union of the polygons corresponding to its vertices.
        We show that our construction indeed yields the stated bounds~on~$\alpha$,~$\beta$,~and~$\gamma$:
        \begin{description}[nosep]
        \item[$\alpha$:] As we consider an $(18\Lambda r)$-tree partition, each element of $\J$ has diameter at most $18\Lambda r$.

        \item[$\beta$:] Recall that $\beta$ is a lower bound on the geodesic distance between any two junctions. First, observe that the shortest path in $\T$ between any two marked vertices has length at least~$\lceil \Lambda\rceil$: this holds by construction for marked vertices on the same heavy path.
        Otherwise, the path contains a light edge, and descending from its lower vertex to the marked vertex traverses a prefix of a heavy path of length at least $\lceil \Lambda\rceil$ by how we marked the vertices.
        Finally, as we consider an \emph{$18$-distance-preserving} $(18 \Lambda r)$-tree partition, a path of length $\lceil \Lambda\rceil$ in $\mathcal{T}$ induces a distance of at least $\Lambda r (\lceil \Lambda\rceil-2)$ between two points in different junctions. As~$\Lambda \geq 6$, we obtain that~$\beta \geq \Lambda(\lceil \Lambda\rceil-2)r \geq \Lambda(2\Lambda/3)r = \frac{2}{3}\Lambda^2r$.

        \item[$\gamma$:] To bound the diameter of a face in $\mathcal{F}$, we first bound the diameter of its connected component in~$\mathcal{T}$. By the properties of the heavy-light decomposition, any shortest path in~$\mathcal{T}$ contains at most $2\log  |\mathcal{T}|$ light edges, where~$|\mathcal{T}|$ denotes the number of vertices of $\mathcal{T}$.
        Thus, such a path traverses at most $2\log|\mathcal T|+1$ heavy-path segments.
        Because there are at most $2n-1$ pieces in an ($18 \Lambda r$)-tree partition according to~\cref{thm:tree-partition}, we have (for $n \geq 8$) that~${2\log |\mathcal{T}| +1 \leq 2 \log (2n-1) + 1 \leq 3 \log n}$.

        \hspace*{15pt} As every $\lceil \Lambda\rceil$'th vertex on a heavy path is marked, a shortest path containing no marked vertices visits at most $(\lceil \Lambda\rceil-1) \cdot (3\log n)$ vertices.
        As we consider an $(18\Lambda r)$-tree partition, every piece has diameter at most $18\Lambda r$.
        Hence, the diameter of any face in~$\mathcal{F}$ is bounded by~${18\Lambda r \cdot (\lceil \Lambda\rceil-1) \cdot (3\log  n) \leq (54 \Lambda^2 \log  n) r}$, i.e., $\gamma \leq (54\Lambda^2 \log  n) r$.
        \end{description}
        The tree partition can be computed in $O(n\log n)$ time by \Cref{thm:tree-partition}. The heavy-light decomposition, the marking, and the computation of the connected components of unmarked vertices take linear time in the size of $\T$. Thus the running time is~$O(n\log n)$.
    \end{proof}

    With slight abuse of notation, we simplify notation to be more concise in the remainder:
    \begin{definition}
        \label[definition]{def:the-junction-tree}
        Given a polygon~$P$ in the $r$-support of a $\Lambda$-thin point set with $\Lambda \geq 6$, we define \emph{the junction tree of $P$} to be the $(18\Lambda r,\frac{2}{3} \Lambda^2r,(54\Lambda^2\log n)r)$-junction tree from \cref{thm:thm-junction-tree}.
    \end{definition}

    \subsection{Comb obstructions}\label{sec:combs}

    The main theorem of this section, Theorem~\ref{lem:the-hard-obstruction}, is the comb obstruction.
    Before we can state it, we require additional definitions and tools.
    See~\cref{fig:def_spine_nerve} for an illustration of the following.

    \begin{figure}
        \centering%
        \includegraphics{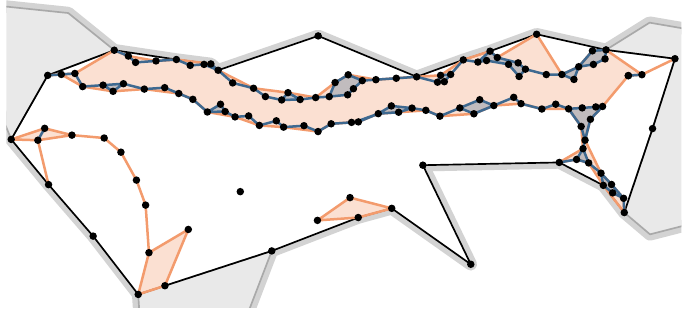}
        \caption{A face $F$ (white) that is adjacent to three junctions (grey).
        The face contains three $F$-spines~(orange) and four $F$-nerves (blue) that are not singletons.}
        \label{fig:def_spine_nerve}
    \end{figure}

    \begin{definition}[Spines and nerves]
    Let $P$ be a weakly simple polygon in $\supp{r}$. A $P$-spine is a connected component of $\supp{\Lambda^{-1}r}$ that lies in $P$. For a junction tree $(\J,\mathcal{F})$ of~$P$ and any face $F\in\mathcal{F}$, an $F$-spine is a connected component of $\supp{\Lambda^{-1}r}\cap F$.
    We similarly define a \emph{$P$-nerve} as a connected component of $\supp{\Lambda^{-2}r}\cap P$, and an \emph{$F$}-nerve as a connected component of $\supp{\Lambda^{-2}r}\cap F$.
    \end{definition}

    Observe that, for a polygon $P\in\supp{r}$ and a face $F$ of the junction-tree of $P$, every $P$-nerve (resp. $F$-nerve) is a subset of exactly one $P$-spine (resp. $F$-spine). On the other hand, multiple $F$-nerves (resp. $F$-spines) can lie in the same $P$-nerve (resp. $P$-spine).

    The intuitive idea behind nerves and spines is the following: Consider a sequence of alternating junctions and faces $J_1,F_1,J_2,\ldots,F_{m-1},J_m$ in the junction tree of $P$, and two $P$-spines~$\Sigma_1$ and~$\Sigma_2$ that span this sequence.
    We would like to connect all points of \emph{both} $\spine_1$ and~$\spine_2$ independently.
    But as they both span the sequence $J_1,F_1,\ldots,F_{m-1},J_m$, they are very long, and connecting them independently of each other results in high dilation (recall \cref{sec:tech_overview}, \mbox{\cref{fig:intuition} (a)} and (b)).
    Hence, we need to make decisions on where to cut these spines.
    One could, for example, cut $\spine_1$ into smaller pieces by cutting it in every face $F_i$, and then connect these smaller pieces to $\spine_2$ in an attempt to obtain a tree with small dilation.
    This works as long as any pair in $\spine_1$ separated by such a cut is not too close, i.e., not in the same \emph{$P$-nerve}.
    Otherwise the detour in the tree for such a close pair might be too large.
    However, when $\spine_1$ only has close pairs, we need to cut $\spine_2$ instead. The question hence is: \emph{Does every face in the junction tree permit such a cut?}

    The difficulty of this question stems from potentially complex interactions among the whole set of $P$-spines.
    To make these interactions easier to handle, we capture the interactions in a single face~$F$ in a structure that we call the \emph{spine-cutting graph}.
    The spine-cutting graph can be viewed as a dual graph, with vertices being regions separated by $F$-spines and edges between regions if we can cross an $F$-spine between them without cutting through an $F$-nerve, see~\cref{fig:spine_cutting_graph} for an example.
    Intuitively, if the spine-cutting graph is not sufficiently well-behaved, we can certify a lower bound to the optimum as stated in~\cref{lem:the-hard-obstruction}.
    This step should not be too surprising; however, more surprising is the fact that the spine-cutting graph is \emph{expressive enough}: If we are unable to certify a lower bound, then we can exploit the structure of the spine-cutting graph algorithmically to explicitly build a tree of sublinear dilation.

    Throughout the rest of the paper, we refer to the set of junctions intersecting a specific face $F$ by $\junc(F)$, and we denote the set of all $F$-spines by $\spines{F}$.
    \begin{definition}[Buffer]\label[definition]{def:buffer-zone}
        Let $(\J,\mathcal{F})$ be the junction tree of a weakly simple polygon $P \in \supp{r}$. For a face $F \in \mathcal{F}$ the $z$-\emph{buffer} $\buffer{z}{F}$ consists of all points in $F$ with geodesic distance at most $z$ of any junction in $\junc(F)$. For the fixed face $F$, \emph{the $z$-buffer of junction $J$} is the subset of the buffer that is within geodesic distance $z$ of $J$.
    \end{definition}

    \begin{figure}
        \centering%
        \includegraphics{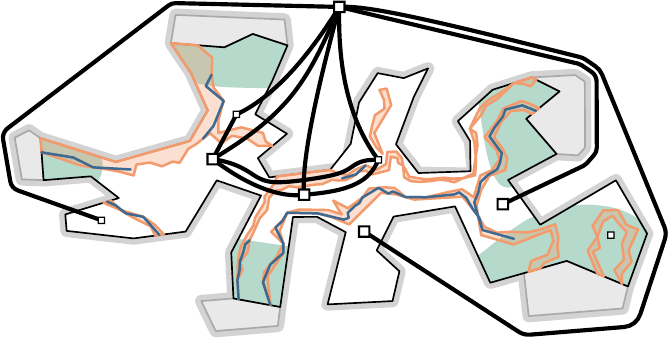}
        \caption{
            The spine-cutting graph~$G_{F}$ of a face $F$ bounded by five junctions, with the buffer shown in green.
            Relevant vertices are large nodes, auxiliary vertices are small.
        }
        \label{fig:spine_cutting_graph}
    \end{figure}

    We can now define the spine-cutting graph.
    \begin{restatable}[{Spine-cutting graph}]{definition}{defSpineCuttingGraph}
    Let $F$ be a face in the junction tree of $P \in \supp r$, and let $\spineset\subseteq\spines{F}$ be a set of $F$-spines. We insert a vertex into the \emph{spine-cutting graph} $G_{F,\spineset}$ for each
    \begin{itemize}[nolistsep,noitemsep]
        \item connected component of $F \setminus \spineset$, and
        \item the outer face given by $\mathbb{R}^2\setminus(\int(F)\cup\junc(F))$.
    \end{itemize}
    We insert an edge into $G_{F,\spineset}$ for any two regions that have a path in $F$ between them that
    \begin{itemize}[nolistsep,noitemsep]
        \item intersects at most one $F$-spine in $\spineset$,
        \item intersects no $F$-nerve contained in an $F$-spine of $\spineset$,
        \item and remains outside the buffer $\buffer{z_1}{F}$, where $z_1 := 4r$.
    \end{itemize}
    \end{restatable}

    \begin{definition}\label[definition]{def:relevant_vertices}
        The \emph{relevant} vertices of the spine-cutting graph $G_{F,\spineset}$ are the outer-face vertex and those whose corresponding region intersects at least two junctions.
        All other vertices are \emph{auxiliary}.
    \end{definition}

    \begin{theorem}[Comb obstruction]\label[theorem]{lem:the-hard-obstruction}
        Let $S$ be a $\Lambda$-thin point set with $\Lambda \geq 6$,
        let $F$ be a face in the junction tree of $P \in \supp r$, and let $\spineset\subseteq\spines{F}$ be a set of $F$-spines.
        If the relevant vertices of the spine-cutting graph $G_{F,\spineset}$ are not in a single connected component of $G_{F,\spineset}$, then
        \[\opt\geq\frac{1}{100}\Lambda.\]
    \end{theorem}

    In order to facilitate understanding, we first give an intuitive description of the proof.
    Consider the connected component of the spine-cutting graph that does not contain the outer-face vertex.
    This component has to be surrounded by a closed walk $W$ alternating between $F$-nerves and buffers of different junctions.
    We can map this walk to a similar walk in the minimum dilation tree, where the paths connecting the buffers are vertex disjoint and are contained in the face $F$.
    Recall that the geodesic distance between two junctions is bounded from below by $\frac{2}{3} \Lambda^2 r$. Thus, the length of such a tree path connecting the buffers of two different junctions is bounded from below by $\frac{2}{3} \Lambda^2 r - 2z_1$.
    By an edge orientation argument, we can show that the total length of the tree walk inside the junctions is at least the total length of the tree walk between the buffers.
    However, the diameter of a junction is at most $18 \Lambda r$.
    Hence, the dilation is at least~$\frac{\frac{2}{3} \Lambda^2 r - 2z_1}{18 \Lambda r} = \frac{\frac{2}{3} \Lambda^2 r - 8r}{18 \Lambda r} = \frac{1}{27}\Lambda - \frac{4}{9\Lambda}\geq \frac{1}{27}\Lambda - \frac{1}{81}\Lambda$. In reality, the bound is slightly worse;
    the walk in the minimum dilation tree may not be fully contained in, but close to the face~$F$.

    \begin{proof}
        Suppose, for a contradiction, that there is a spanning tree $T$ of $S$ with $\dil(T)<\frac{1}{100}\Lambda$. Set~$\eta\coloneqq\dil(T)/\Lambda$. Thus, $\eta<1/100$.
        Since the relevant vertices of the spine-cutting graph $G_{F,\spineset}$ are not in the same connected component, there is a connected component $C$ of $G_{F,\spineset}$ containing a relevant vertex that does not contain the outer-face vertex. See \Cref{fig:comb_lower_bound}(a) for an example.

        \begin{claim}\label[claim]{claim:closed_walk}
            There is a closed walk in $F$ alternating between paths inside $z_1$-buffers of distinct junctions and portions of $F$-nerve boundaries outside the buffer interiors, using nerves from at least two distinct $F$-spines in $\spineset$.
            It separates a point $p\in F\setminus(\spineset\cup\buffer{z_1}{F})$ from the outer face $\mathbb{R}^2\setminus(\int(F)\cup\junc(F))$.
        \end{claim}

        \begin{claimproof}

            Choose a relevant vertex $v\in C$, and let~$R$ be the corresponding region of $F\setminus\spineset$.
            Since~$R$ meets two distinct junctions, choose a simple arc $\gamma'$ in $R$ joining them.
            Let $\gamma$ be a subarc of~$\gamma'$ connecting two distinct $z_1$-buffers $B_a,B_b$ whose interior avoids all such buffers.
            Note that~$\gamma$ exists and is non-empty, as the pairwise distance of junctions is at least $\frac{2}{3}\Lambda^2r$, which is greater than~$z_1=4r$ for~$\Lambda\geq 6$.
            Let~$K$ be the union of~$\buffer{z_1}{F}$ and the boundaries of the $F$-nerves contained in~$\spineset$.
            Any path from $\gamma$ to the outer face avoiding $K$ would induce a walk in $G_{F,\spineset}$ from~$v$ to the outer-face vertex, a contradiction.
            In particular, restricting that path to the $F$-spines~$\spineset$ yields a witness path for every edge of such a walk in $G_{F,\spineset}$.

            On either side of $\gamma'$, trace the boundary of the connected component of $F\setminus K$ that contains~$\gamma$;
            together, these two walks separate $\gamma$ from the outer face.
            For either side, this gives a path in $K$ from $B_a$ to $B_b$, consisting of $F$-nerve boundaries outside the buffers. Whenever such a path visits a buffer more than once, shortcut between its first and last visit inside that buffer.
            Since $\gamma$ avoids the buffer interiors, these shortcuts do not cross $\gamma$, and the retained $F$-nerve portions remain on their original side of $\gamma'$.
            Thus, each path visits every buffer at most once. Moreover, the paths share no buffer other than $B_a$ and $B_b$, since all other buffers are connected and pairwise disjoint from each other and~$\gamma\cup B_a\cup B_b$.

            Joining the two paths inside $B_a$ and $B_b$ gives a simple closed walk $W$ over distinct buffers that encloses $\gamma$. Choose $p$ in the interior of $\gamma$. Since $\gamma\subset R$ lies outside the buffers, we have~${p\notin\spineset\cup\buffer{z_1}{F}}$, and $W$ separates $p$ from the outer face. Finally, both paths contain $F$-nerve portions, which belong to distinct $F$-spines since they lie on opposite sides of, and are disjoint from $\gamma'$.
        \end{claimproof}

        This walk gives rise to a cyclic sequence of buffers and $F$-nerves in the order that it visits them.
        We denote this sequence by~$B_0,N_0,B_1,N_1,\ldots, N_{k-1}, B_{k}$, where each $B_i$ is the buffer of a distinct junction~$J_i$ except for~$B_0=B_k$, and $N_i$ is an $F$-nerve connecting $B_i$ to $B_{i+1}$,
        see \Cref{fig:comb_lower_bound}(a).
        Note that the $F$-nerves $N_i$ do not have to be distinct, see for example $N_0 = N_1$ in \Cref{fig:comb_lower_bound}(a).
        For an $F$-nerve~$N_i$, we call the part of $\partial N_i$ connecting~$B_i$ and $B_{i+1}$ that is reachable from $p$ without crossing $B_0\cup N_0 \cup B_1 \cup \ldots \cup N_{k-1} \cup B_{k}$ the \emph{inner boundary} of $N_i$, see \Cref{fig:comb_lower_bound}(b).

        \begin{figure}
            \centering
            \includegraphics{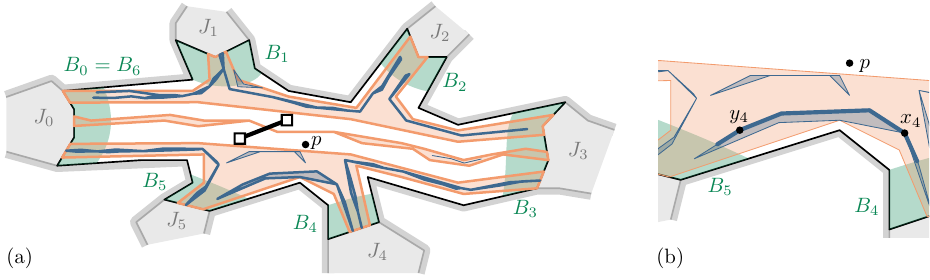}
            \caption{(a) The component of $G_{F,\spines{F}}$ that is separated from the outer-face with a point $p$ inside gives rise to the cyclic sequence $B_0,N_0,B_1,N_1,\ldots, N_{5}, B_{6}$ alternating between distinct buffers and (thick) $F$-nerves connecting them. (b) A more detailed view of $N_4$. The thick boundary is the inner boundary of $N_4$, which defines the points $x_4$ and $y_4$.}
            \label{fig:comb_lower_bound}
        \end{figure}

        \mysubpara{Defining paths between junctions.}
        For an $F$-nerve $N_i$, let $x_i,y_i \in S$ be the first and last point of $S$ encountered when following the inner boundary of $N_i$ from $B_i$ to~$B_{i+1}$, see \Cref{fig:comb_lower_bound}(b). Recall that the boundary of $N_i$ is a subset of the boundary of the $\Lambda^{-2}r$-support, whose vertices are points in $S$ and edges have length at most $\Lambda^{-2}r$. As the length of the inner boundary is at least $\frac{2}{3}\Lambda^2r - 2z_1 = \frac{2}{3}\Lambda^2r - 8r > 2\Lambda^{-2}r$ by the property of the junction tree that the geodesic distance between any two junctions is at least $\frac{2}{3}\Lambda^2r$, we conclude that the points $x_i$ and $y_i$ are well-defined. Furthermore, we have that $\dist_P(x_i,J_i) \leq z_1 + \Lambda^{-2}r$ and $\dist_P(y_i, J_{i+1}) \leq z_1 + \Lambda^{-2} r$.

        Let $L_i\coloneqq \pi_T(x_i,y_i)$ be the shortest path in $T$ between $x_i$ and $y_i$.
        \begin{claim}\label[claim]{claim:in_neighbourhood}
            The vertices of the path $L_i$ are contained in the $\eta r/(2\Lambda)$-neighbourhood of the vertices of the inner boundary of $N_i$ and thereby contained in $P$.
        \end{claim}
        \begin{claimproof}
            Consider the sequence of points of $S$ visited by following the inner boundary of~$N_i$ from $x_i$ to $y_i$.
            Consider for each consecutive pair of points $u,v$ in this sequence the unique path~$\pi_T(u,v)$ in~$T$.
            The union of these paths gives us \emph{a} path in $T$ connecting $x_i$ and $y_i$, and therefore contains the shortest path $L_i\coloneqq \pi_T(x_i,y_i)$.
            Because $\dil(T) = \eta \Lambda$, it follows that for every consecutive pair $u,v$ it holds that $\dist_T(u,v)\leq\dil(T)\|u-v\|\leq \eta r/\Lambda$.
            Every point of~$\pi_T(u,v)$ is therefore within Euclidean distance $\eta r/(2\Lambda)$ of one of its endpoints. Consequently,~$L_i$ is contained in the $\eta r/(2\Lambda)$-neighbourhood of the vertices of $N_i \cap (F \setminus \buffer{z_1}{F})$.
            Since the Delaunay triangulation is a 2-spanner, every vertex of $L_i$ can be joined to a vertex of the inner boundary of $N_i$ by a Delaunay path of length at most $\eta r/\Lambda<r$. All edges of this path have length less than $r$, so its endpoints lie in the same polygon of the $r$-support, and hence every vertex of $L_i$ lies in $P$.
        \end{claimproof}

        \mysubpara{Making the paths vertex-disjoint.}
        We would like the paths $L_i$ to be vertex-disjoint.
        To this end, we first prove that any two paths $L_i,L_j$, where the corresponding $F$-nerves $N_i,N_j$ are contained in different $F$-spines, are vertex-disjoint.
        Second, we show how to shortcut any intersecting paths such that we end up with vertex-disjoint paths that still connect buffers in a cyclic manner around~$p$.

        Let $L_i$ and $L_j$ be two paths for which $N_i$ and $N_j$ are contained in different $F$-spines $\Sigma_i$ and~$\Sigma_j$. For contradiction, assume that there is a vertex $v$ that is both on $L_i$ and $L_j$.
        By~\Cref{claim:in_neighbourhood},~$v$ lies within Euclidean distance $\eta r/(2\Lambda)$ of a vertex $v_i$ of the inner boundary of $N_i$ and a vertex~$v_j$ of the inner boundary of $N_j$.
        So, $\| v_i - v_j\| \leq \eta r/\Lambda$.
        As the Delaunay triangulation is a~\mbox{2-spanner}~\cite{DBLP:journals/siamcomp/Xia13}, there is a path in the Delaunay triangulation of $S$ between $v_i$ and $v_j$ of length at most $2\eta r/\Lambda$. In particular, the maximum edge length of this path is $2\eta r/\Lambda < r/\Lambda$ and thus this path is contained in~$\supp{\Lambda^{-1}r}$. However,~$v_i$ and~$v_j$ being outside of the buffer and them being in distinct $F$-spines implies that $\dist_{\supp{\Lambda^{-1}r}}(v_i,v_j) \geq 2z_1$, a contradiction.

        Next, our goal is to find a sequence of \emph{vertex-disjoint} paths $L'_0, \ldots, L'_{\ell-1}$ with $\bigcup_j L'_j \subseteq \bigcup_i L_i$ that again connects\footnote{Strictly speaking, these paths do not connect buffers, but connect points close to buffers, i.e., their representative points $x_i$ and $y_i$ are chosen as endpoints of the paths $L_i$.} (a subset of) the buffers $B_0,\ldots, B_k$ in a cyclic manner.
        By \Cref{claim:closed_walk}, there are at least two distinct $F$-spines that contain the $F$-nerves $N_i$. We therefore assume without loss of generality that $N_0$ and $N_{k-1}$ are contained in different $F$-spines. We construct this sequence of paths~$L'_0, \ldots, L'_{\ell-1}$ as follows. \Cref{fig:new_paths} illustrates the approach.

        \begin{figure}
            \centering
            \includegraphics{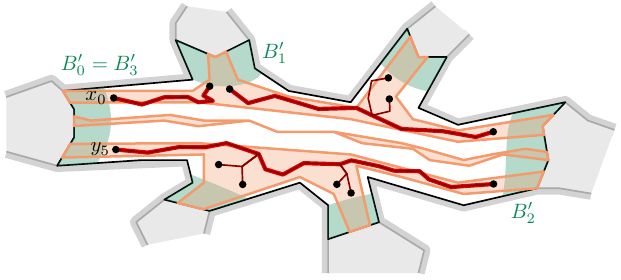}
            \caption{The original paths $L_0,\ldots,L_5$ (thin darkred) are replaced by the paths $L'_0,L'_1,L'_2$ (thick red) that are vertex-disjoint connecting the buffers $B'_0,B'_1,B'_2, B'_3 = B'_0$. }
            \label{fig:new_paths}
        \end{figure}

        We process $L_0,\ldots,L_{k-1}$ in order, maintaining a sequence of pairwise vertex-disjoint tree paths~$L'_j$. If $L_i$ is vertex-disjoint from every path $L'_j$ obtained so far, append $L_i$ to the sequence. Otherwise, let $L'_a$ be the first path in the sequence obtained so far that intersects~$L_i$. Delete every path after~$L'_a$ from the sequence, and replace $L'_a$ by the unique path in $T$ joining the starting point of $L'_a$ to the endpoint of $L_i$, i.e., to $y_i$. Since $T$ is a tree, this replacement path is contained in $L'_a\cup L_i$. It is therefore disjoint from paths $L'_b$ for $b < a$. We conclude that the sequence $L'_0,\ldots,L'_{\ell-1}$ of paths obtained this way is indeed contained in $\bigcup_i L_i$, and therefore~\Cref{claim:in_neighbourhood} still holds for these paths.

        It remains to argue that $\ell\geq 2$, i.e., that $L'_0$ is never merged with $L_{k-1}$. By induction over the procedure, every path $L'_a$ in the current sequence is contained in the union of original paths $L_i$ whose nerves $N_i$ lie in a single $F$-spine: a replacement path lies in $L'_a\cup L_i$, and $L_i$ can intersect~$L'_a$ only if $N_i$ lies in the same $F$-spine as the nerves defining $L'_a$, by the vertex-disjointness shown above. In particular, $L'_0\subseteq\bigcup\{L_i : N_i\subseteq\Sigma_0\}$, where $\Sigma_0$ is the $F$-spine containing $N_0$. As $N_{k-1}$ lies in a different $F$-spine, $L_{k-1}$ is vertex-disjoint from every such $L_i$ and hence from $L'_0$. It follows that $\ell\geq 2$, the start point of $L'_0$ is $x_0$, the endpoint of $L'_{\ell-1}$ is $y_{k-1}$, and consecutive paths connect buffers of \emph{distinct} junctions in the required cyclic way.

        This gives rise to a new cyclic sequence $B'_0,L'_0,B'_1,\ldots,L'_{\ell-1},B'_{\ell}$. Note that $B'_0 = B'_{\ell} = B_0$ and that each $B'_j$ must also appear in the original sequence as some $B_i$. With slight abuse of notation, we use $x_j$ and $y_j$ to denote the start- and endpoint of $L'_j$ from here on.

        \mysubpara{Orienting the tree.}
        Let $Q_j\coloneqq \pi_T(y_j,x_{j+1})$ be the path in $T$ connecting the endpoint of~$L'_j$ and the start point of $L'_{j+1}$.
        Consider the closed walk $W=L'_0\oplus Q_0\oplus\cdots\oplus L'_{\ell-1}\oplus Q_{\ell-1}$ in~$T$ {that is the concatenation of these paths}.
        We orient the tree $T$ as follows. Orient every edge along $L'_j$ from $x_j$ to $y_j$ according to that path, and orient all remaining edges of $T$ arbitrarily. Note that this indeed gives a valid oriented tree, as the paths $L'_j$ are disjoint. For every edge $e$ of $T$, let $n^+(e)$ and~$n^-(e)$ be the number of times the walk $W$ traverses $e$ with and against its orientation. Removing an edge $e$ from the tree $T$ separates it into two components. Since $W$ is a closed walk, it crosses the resulting cut equally often in both directions. Hence, we have~$n^+(e)=n^-(e)$ for every edge $e$. The paths $L'_j$ are pairwise edge-disjoint and are traversed according to their orientations, whereas the signed contribution of the paths $Q_j$ is bounded from below by minus their total length. Therefore\vspace{-2mm}
        \[0=\sum_{e\in T}\len(e)\bigl(n^+(e)-n^-(e)\bigr)\geq\sum_{j=0}^{\ell-1}\len(L'_j)-\sum_{j=0}^{\ell-1}\len(Q_j).\]

        \mysubpara{A lower bound for \boldmath$\len(L'_j)$.}
        To provide a lower bound on $\len(L'_j)$, we would like to make use of the fact that the geodesic distance between any two junctions is at least $\frac{2}{3}{\Lambda^2r}$. However, as only the vertices of $L'_j$, and not necessarily its edges, are contained in $P$, we cannot apply this bound directly. Instead, we show that each edge of $L'_j$ can be replaced by a path contained in $P$ that is at most twice as long.
        Consider an edge of $L'_j$ between two points $u,v \in S$ and the shortest path~{$\pi_\mathrm{Del}(u,v)$}
        in the Delaunay triangulation of $S$ between $u$ and~$v$.
        Since $L'_j\subseteq\bigcup_i L_i$, the edge $uv$ lies on the tree path
        $\pi_T(u',v')$ between two consecutive vertices $u',v'$ of the inner boundary of some
        $F$-nerve $N_i$, and $\|u'-v'\|\leq\Lambda^{-2}r$ (\Cref{claim:in_neighbourhood}).
        Hence $\|u-v\|=\dist_T(u,v)\leq \dist_T(u',v')\leq\dil(T)\|u'-v'\|\leq \eta r/\Lambda$.
        The shortest path~{$\pi_\mathrm{Del}(u,v)$}
        has length at most $2\|u-v\|$, as the Delaunay triangulation is a 2-spanner~\cite{DBLP:journals/siamcomp/Xia13}.
        This bound also applies to each edge of {$\pi_\mathrm{Del}(u,v)$},
        and thus the length of each edge is bounded by $2 \eta r/\Lambda \leq r$.
        It follows that~{$\pi_\mathrm{Del}(u,v)$}
        is contained in the connected component of $\supp{r}$ that contains $u$ and $v$, which is $P$.

        Using this bound, together with the fact that $L'_j$ starts within $z_1 + r/\Lambda^2$ geodesic distance of one junction and ends within~$z_1 + r/\Lambda^2$ geodesic distance of another junction, we conclude that~${\len(L'_j)\geq \frac{1}{2}\dist_P(x_j,y_j) \geq \frac{1}{3}\Lambda^2r - z_1 - r/\Lambda^2}$. So,  $\sum_{j=0}^{\ell-1}\len(L'_j) \geq \ell(\frac{1}{3}\Lambda^2r- z_1- r/\Lambda^2)$.

        \mysubpara{Putting everything together.}
        It follows that $\sum_j\len(Q_j)\geq\sum_j\len(L'_j) \geq\ell(\frac{1}{3}\Lambda^2r -z_1 - r/\Lambda^2)$. Hence, by the pigeon-hole principle, some $Q_j$ has length at least $\frac{1}{3}\Lambda^2r- z_1- r/\Lambda^2$. The path $Q_j$ connects two points $y_j$ and~$x_{j+1}$ within $z_1 + r/\Lambda^2$ geodesic distance of the same junction. As the geodesic diameter of a junction is at most $18\Lambda r$, we have $\|y_j-x_{j+1}\|\leq18\Lambda r + 2z_1 + 2r/\Lambda^2$. Consequently, $\len(Q_j)< \eta\Lambda(18\Lambda r + 2z_1+ 2r/\Lambda^2)$.
        Using that $\Lambda \geq 6$ and $\eta < \frac{1}{100}$, we find
        \begin{equation*}
            \len(Q_j)< \eta\Lambda(18\Lambda r + 2z_1 + 2r/\Lambda^2) < \frac{18}{100} \Lambda^2r + \frac{8}{100}\Lambda r + \frac{2}{100} r/\Lambda
            < \frac{1}{3}\Lambda^2 r-z_1- r /\Lambda^2,
        \end{equation*}
        which is a contradiction.
        Thus, every spanning tree of $S$ has dilation at least $\frac{1}{100}\Lambda$.
    \end{proof}

    \section{Structuring the point set via the spine-cutting graph}\label{sec:consistent}

    In this section, we use the spine-cutting graph to compute structures that dictate how our bottom-up algorithm should construct a tree of sublinear dilation.

    \subsection{Realizing the spine-cutting graph}

    \begin{restatable}[{Good cut}]{definition}{defGoodCut}\label[definition]{def:minimal_realization}
        Let $F$ be a face of a junction tree of~${P\in\supp{r}}$, and let $\spineset\subseteq\spines{F}$ be a set of $F$-spines.
        A \emph{good cut} of the spine-cutting graph $G_{F,\spineset}$ is a planar straight-line embedded graph~$C$ for which the following conditions hold:
        \begin{enumerate}[nolistsep]
    \item\label{itm:in_spines} $C$ is contained in the $F$-spines in $\spineset$ that touch junctions of $F$.
    \item\label{itm:not_intersect_nerve} $C$ does not intersect the $z_2$-buffer or any $F$-nerves, where $z_2 := r/2$.
    \item \label{itm:tree}
    The union of $C$ and all vertex regions has one connected component that contains the regions of all relevant vertices and all of $C$.

    \item\label{itm:minimal} $C$ does not separate $F$-spines from junctions, i.e., for every $\spine\in\spineset$ that intersects at least one junction, every connected component of $\spine\setminus C$ also intersects at least one junction.
    \item\label{itm:path_junctions} $C$ does not separate junctions connected by $F$-spines from each other, i.e., for every pair of junctions $J_1,J_2 \in \junc(F)$, for which there is a path from $J_1$ to $J_2$ in $\spineset \cup \junc(F)$, there is also a path in $(\spineset \cup \junc(F)) \setminus C$ connecting $J_1$ and $J_2$.
        \end{enumerate}
    \end{restatable}
    Good cuts are \emph{the central tool} of our algorithm. In spirit, a good cut $C$ is a set of cuts of $F$-spines in $\spineset$ that gets rid of all cycles in the polygonal domain $\spineset\cup\junc(F)$ that use more than one junction of $\junc(F)$.

    \begin{restatable}{lemma}{lemmaMaxEdgeLengthSupport}\label[lemma]{lem:max_edge_length_support}
        Let $S$ be a $\Lambda$-thin point set with $\Lambda \geq 6$, and let $P$ be a weakly simple polygon in $\supp{r}(S)$.
        Then any edge of the Delaunay triangulation of $S$ contained in $P$ has length at most $3\Lambda r$.
    \end{restatable}

    \begin{proof}[Proof sketch]
        We show that a long Delaunay edge $e$ contained in $P$ implies that $P$ contains a large empty disk; in particular, the empty circumcircle of a Delaunay triangle using $e$.
        By definition, the boundary $\partial P$ consists of Delaunay edges of length at most $r$.
        Hence, $\partial P$ cannot ``poke very far into'' such an empty circumcircle if~${\len(e)\geq 3\Lambda r}$.
        This gives rise to a concentric disk of radius at least $\Lambda r$ that is fully contained in $P$.
        By~\cref{def:thinness}, $P$ is thus not $\Lambda r$-thin.
    \end{proof}

    \begin{figure}%
        \centering%
        \includegraphics{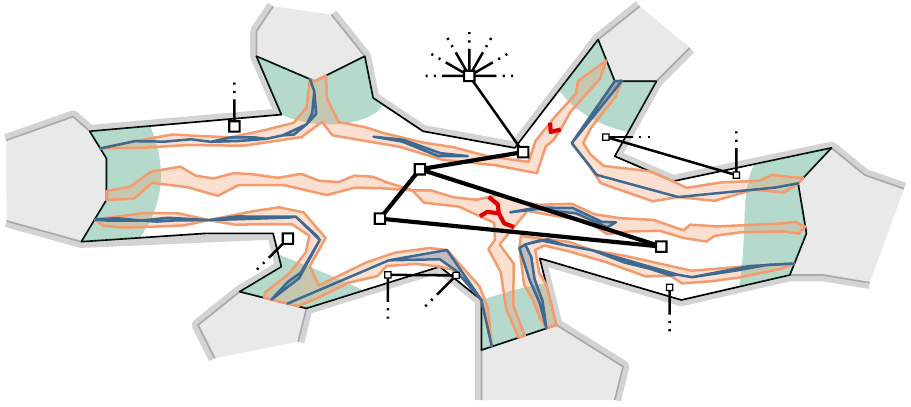}%
        \caption{The red planar straight-line graph is a good cut $C$ of the spine-cutting graph $G_{F,\spines{F}}$. The subgraph of $G_{F,\spines{F}}$ induced by $C$ includes all bold edges. Dotted edges connect to the outer face vertex.
        Observe that the left and right junctions remain connected via the bottom two junctions and their connecting spines.}%
        \label{fig:minimal_realization}%
    \end{figure}

    \begin{restatable}{lemma}{lemmaFindCutOrLowerBound}\label[lemma]{lem:construct_min_realization}
        Let $S$ be a $\Lambda$-thin point set with $\Lambda \geq 6$, and $(\J,\mathcal{F})$ be the
        junction tree of a weakly simple polygon~${P \in \supp{r}}$. Given a face $F\in\mathcal{F}$, and a subset $\spineset\subseteq\spines{F}$, we can construct a good cut of the spine-cutting graph $G_{F,\spineset}$, or certify that $\opt\geq\frac{1}{100}\Lambda$ in $O(n^2\log n)$ time.
    \end{restatable}

    \begin{proof}[Proof sketch.]
    Recall that each relevant vertex of~$G_{F,\spineset}$ corresponds to a region; either to the outer face or to a connected component of~$F\setminus\spineset$ that intersects at least two junctions.
    A good cut of~$G_{F,\spineset}$ must connect all these regions (property~\ref{itm:tree}).
    The high-level idea is to repeatedly look for shortest paths from the outer face through the Delaunay triangulation of the $F$-spines in~$\spineset$ that reach a new region corresponding to a relevant vertex.
    If we find suitable such paths for all relevant vertices, then their union forms a tree from which we derive a good cut of~$G_{F,\spineset}$.

    We construct an auxiliary graph $H$, \cref{fig:spine_cutting} illustrates an example.
    The vertices of~$H$ are those of $G_{F,\spineset}$ together with a vertex for each triangle of the Delaunay triangulation of~$S$ that meets three criteria:
    it (i) is contained in an $F$-spine of $\spineset$,
    (ii)~is interior-disjoint from all $F$-nerves, and
    (iii)~does not intersect the $z_2$-buffer of~$F$ for~${z_2 = r/2}$.
    As an $F$-spine is only weakly simple, part of its boundary may form polygonal paths.
    We add a vertex to the graph~$H$ for each edge in such a path that has these properties (ii) and (iii).
    A triangle (or degenerate boundary segment) vertex in $H$ is adjacent to all other regions (including triangles) with which it shares a segment on a Delaunay edge outside of $F$-nerves and the $z_2$-buffer of $F$.
    The outer-face vertex is initially adjacent to all regions of vertices in $G_{F,\spineset}$ that touch it, but we immediately contract such edges.

    To find a tree $T_H$, we repeatedly select shortest paths in $H$ from the outer face vertex to an unreached relevant vertex, add it to $T_H$ and contract its vertices in $H$.
    If no such path exists, we show that the relevant vertices are disconnected in $G_{F,\spineset}$.
    Each edge of the spine-cutting graph has a witness path $\pi$ that avoids the interior of nerves and the $z_1$-buffer.
    Every Delaunay triangle intersected by $\pi$ has diameter at most~$3r$ by~\Cref{lem:max_edge_length_support}, and hence does not intersect the $z_2$-buffer of~$F$, as we have $z_1 - z_2 =4r- r/2 > 3r$.
    It follows that any path in the spine-cutting graph has a corresponding path in $H$ (and we obtain properties~\ref{itm:not_intersect_nerve} and~\ref{itm:tree}).

    We can finally show that a plane, acyclic embedding of~$T_H$, where we connect the midpoints of Delaunay edges and retain only the parts contained in $F$-spines of $\spineset$ that touch junctions (property~\ref{itm:in_spines}), yields a good cut $C$.
    It remains to show two properties:
    if~cutting $\spineset\cup\junc(F)$ by~$C$ splits a piece of an $F$-spine from all junctions that spine touched (thus violating property~\ref{itm:minimal}), or separates two junctions that were connected by spines (thus violating property~\ref{itm:path_junctions}), we can show that $T_H$ must contain a cycle, a contradiction.
    \end{proof}
    \begin{figure}
        \centering%
        \includegraphics{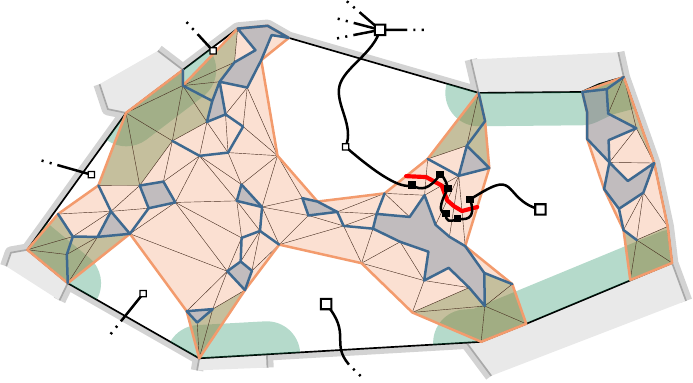}%
        \caption{
            A tree in the graph $H$, derived from two $F$-spines $\spineset$ (orange), that connects the three relevant vertices (large white nodes) over auxiliary vertices (small white nodes) and triangles of the Delaunay triangulation (small black nodes).
            The resulting good cut $C$ is shown in red.
        }
        \label{fig:spine_cutting}
    \end{figure}

    We need the cuts of individual faces to play nicely with each other, to enable inductive proofs later on. For this, we define the following notion. Recall that a $P$-spine is a weakly simple polygon in $\supp{\Lambda^{-1}r}$ that is contained in a polygon $P$.

    \begin{definition}[Well-spaced]\label[definition]{def:well-spaced-cut}

        Let $(\J,\F)$ be the junction tree of a weakly simple polygon~${P \in \supp{r}}$ and let $\spineset$ be a set of $P$-spines.
        Let $\mathcal{C}$ be a union of good cuts of different spine-cutting graphs $G_{F,\spineset_F}$ of faces $F \in \F$, where $\spineset_F$ is the set of $F$-spines induced by the $P$-spines in $\spineset$.
        We say that~$\mathcal{C}$ is \emph{well-spaced} w.r.t.\ $\spineset$ if the following holds for every connected component $\overline{\spine}$ of $\spine\setminus \mathcal{C}$ of every~${\spine\in\spineset}$:

        \begin{itemize}[nolistsep,noitemsep]
        \item Any path $\gamma\subset \overline \spine$ connecting two points $p,q\in S\cap \overline\spine$ can be mapped to a polygonal walk over the vertices $S\cap{\overline \spine}$ with edge length at most $3r$, such that the geodesic distance between a point and its image is at most $3r$.
        \end{itemize}
    \end{definition}

    \begin{lemma}
        \label[lemma]{lem:well-spaced-cuts}
        Let $S$ be a $\Lambda$-thin point set with $\Lambda \geq 6$, and let~$(\J,\mathcal{F})$ be a
        junction tree of a weakly simple polygon $P \in \supp{r}$. Let $\spineset$ be a set of $P$-spines. Let $\mathcal{C}$ be a set consisting of good cuts of spine-cutting graphs~$G_{F,\spineset_F}$ of faces $F \in \F$ constructed using~\Cref{lem:construct_min_realization}, where $\spineset_F$ is the set of $F$-spines induced by the $P$-spines in $\spineset$.
        Then $\mathcal{C}$ is well-spaced w.r.t.\ $\spineset$.

    \end{lemma}
    \begin{proof}
        Let $\spine \subseteq P$ be a weakly simple polygon in $\supp{\Lambda^{-1}r}$ and let $\overline \spine$ be a connected component of~$\spine \setminus \mathcal{C}$.
        Let~$\gamma$ be a curve in $\overline \spine$ connecting two points $p,q\in S\cap \overline \spine$.

        The intersection of any Delaunay triangle $\Delta$ and  $\mathcal{C}$ consists of at most two line segments joining edge midpoints.
        Hence, each connected piece $\overline{\Delta}$ of $\Delta\setminus\mathcal{C}$ contains a vertex of~$\Delta$, and the points in $S\cap \overline{\Delta}$ are connected by Delaunay edges.
        By~\cref{lem:max_edge_length_support}, these edges have length at most~$3r$.
        Whenever~$\gamma$ passes from $\overline{\Delta}$ to another triangle piece $\overline \Delta'$ in $\overline{\spine}$, it crosses a connected piece of a Delaunay edge that has a shared vertex of the two triangles at one endpoint.
        We select this vertex as the exit/entry vertex between $\Delta$ and $\Delta'$, and if $\gamma$ passes directly through a vertex, we simply choose this vertex.
        Degenerate pieces of Delaunay triangles, edges on the boundary of~$\supp{\Lambda^{-1}r}$, are handled analogously using their endpoints.
        Further, we fix the first entry and the last exit vertices to be~$p$ and~$q$, respectively.
        To obtain a polygonal path between the entry and exit vertices of~$\overline \Delta$, we follow the intact Delaunay edges of the triangle piece.
        The concatenation of these paths is our polygonal output walk $\pi$.
        By definition, $\pi$ consists of edges of length at most $3r$ that are contained in~$\overline \spine$.
        We traverse the corresponding pieces of $\gamma$ and $\pi$ simultaneously, waiting at exit vertices.
        As the diameter of each visited Delaunay triangle is at most $3r$, the geodesic distance in~$\overline{\spine}$ between a point on $\gamma$ and the corresponding point on~$\pi$ is then at most $3r$.
    \end{proof}

    \begin{corollary}\label[corollary]{cor:well-spaced-cuts}
        Let $S$ be a $\Lambda$-thin point set with $\Lambda \geq 6$, and let $(\J,\F)$ be a junction tree of a weakly simple polygon $P\in\supp{r}$. Let $\spineset_1,\ldots,\spineset_k$ be pairwise-disjoint sets of $P$-spines.
        Let $\mathcal{C}_1,\ldots,\mathcal{C}_k$ be sets consisting of good cuts of spine-cutting graphs $G_{F,\spineset_{i,F}}$ of faces $F\in\F$ constructed using \cref{lem:construct_min_realization}, where $\spineset_{i,F}$ is the set of $F$-spines induced by the $P$-spines in $\spineset_i$. Then $\mathcal{C}=\bigcup_{i=1}^k\mathcal{C}_i$ is well-spaced w.r.t.\ $\bigcup_{i=1}^k\spineset_i$.
    \end{corollary}
    \begin{proof}
        Well-spacedness is defined on a per-$P$-spine basis. As $\spineset_i$ are pairwise disjoint, and the cuts from $\mathcal{C}_i$ are contained in $P$-spines in $\spineset_i$, this is an immediate consequence of \cref{lem:well-spaced-cuts}.
    \end{proof}

    \subsection{A consistent sequence of junction trees}\label{sec:consistent_sequence}

    Our algorithm constructs a tree in a bottom-up fashion by considering edges in $r$-supports for increasing $r$. In particular, we want to consider values of $r$ that increase by a factor $\Lambda$. However, depending on the spread such a sequence might not have length polynomial in $n$.
    Therefore, we construct a sequence of $O(n^2)$ relevant scales $(r_0,\ldots,r_N)$ such that:
    \begin{itemize}
        \item either $\frac{r_i}{r_{i-1}}=\Lambda$ or $\frac{r_i}{r_{i-1}}\geq\Lambda^{10}$, and
        \item for any pair of points $p,q \in S$, $p\neq q$, there is some $i$ for which $r_{i-1}\leq \|p-q\| < r_i=\Lambda r_{i-1}$, and
        \item for any pair of points $p,q \in S$, $p \neq q$, we have that $r_0 <\|p-q\|$.
    \end{itemize}
    We call this a compressed scale sequence. Observe these restrictions imply that if $r_{i+1} \neq \Lambda r_i$, then there is no pair of points in $S$ with pair-wise distance in $[r_i,r_{i+1})$.

    \begin{lemma}\label[lemma]{lem:scale_sequence}
        A compressed scale sequence can be computed in $O(n^2\log n)$ time.
    \end{lemma}
    \begin{proof}
        First compute and sort all pairwise distances between points in $S$ in $O(n^2 \log n)$ time to obtain a sequence~$\delta_0,\ldots,\delta_k$ with $k = O(n^2)$. We start with the sequence $\mathcal{R} := (\delta_0/2, \delta_0 \Lambda/2)$, and then process the distance sequence in order, appending values to $\mathcal{R}$ in the process.
        If the current distance is $\delta_i < r$, where $r$ is the last value in $\mathcal{R}$, continue to the next distance $\delta_{i+1}$. If~$r \leq \delta_i < r\Lambda^{10}$, append $r \Lambda,r \Lambda^2, \ldots, r\Lambda^{10}$ to $\mathcal{R}$. Otherwise, append $\delta_i$ and $\delta_i \Lambda$ to~$\mathcal{R}$.
    \end{proof}

    We conclude by packaging all assumptions we made so far into the following notion.

    \begin{definition}[Consistent sequence]
    Let $\mathcal{R}$ be a compressed scale sequence.
    Consider a sequence that, for every value $r\in\mathcal{R}$, consists of a set of junction trees that contains for every weakly simple polygon $P\in\supp{r}$ a junction tree $(\J_r,\F_r)$ of $P$.
    This sequence is \emph{consistent} if for every junction tree~$(\J_r,\F_r)$ and for every face $F \in \mathcal{F}_r$, there exists a good cut of the spine-cutting graph~$G_{F,\spines{F}}$.
    \end{definition}

    \begin{theorem}\label[theorem]{thm:construct_consistent_sequence}
        Let $S$ be a $\Lambda$-thin set of $n$ points with $\Lambda \geq 6$. We can compute a consistent sequence of junction trees for $S$ or certify that $\opt \geq \frac{1}{100}\Lambda$ in $O(n^5\log n)$ time.
    \end{theorem}
    \begin{proof}
        We first compute a compressed scale sequence in $O(n^2 \log n)$ time (\Cref{lem:scale_sequence}). For each of the $O(n^2)$ values $r$ in the sequence, we construct for each weakly simple polygon~$P$ in the $r$-support~$\supp{r}$ a junction tree of~$P$ in~$O(n\log n)$ time (\Cref{thm:thm-junction-tree}). For each face in the junction tree, we invoke~\Cref{lem:construct_min_realization} to either certify that~$\opt \geq \frac{1}{100}\Lambda$, or
        to construct a good cut of the spine-cutting graph $G_{F,\spines{F}}$ in~$O(n^2\log n)$ time.
        As there are $O(n)$ faces in total in the junction tree for a single value $r$, and there are $O(n^2)$ values of $r$ that are being considered, the total running time is thus $O(n^2\cdot (n^3\log n))=O(n^5\log n)$.
    \end{proof}

    \section{The algorithm}\label{sec:algorithm}
    Using our structural results and definitions from \Cref{sec:obstructions} and \Cref{sec:consistent}, we finally present our main algorithm. Given a consistent sequence of junction trees, it constructs a tree on $S$ with dilation $O(\max\{n/\Lambda,\Lambda^{14}\})$. Overall, this yields an $\tilde{O}(n^{14/15})$-approximation for the minimum dilation tree for any point set in $\mathbb{R}^2$. Recall that we assume $\Lambda\geq 6$.

    \subsection{Description of the algorithm}\label{sec:description_of_algorithm}

    The pseudocode for the algorithm can be found in \Cref{alg:the_alg}.
    The algorithm is parameterized by a scale-separation factor $\Lambda$, which we will later set to approximately $n^{1/15}$.
    We first check whether our lower bounds apply:
    Using~\Cref{lem:thin-obstruction-check}, we can either verify that $S$ is $\Lambda$-thin or certify that~$\opt \geq 2 \Lambda$, and using~\Cref{thm:construct_consistent_sequence} we can then either construct a consistent sequence of junction trees or certify that $\opt \geq \frac{1}{100}\Lambda$.
    In the two cases where we find a lower bound on $\opt$, we return the Euclidean minimum spanning tree. Otherwise, the algorithm proceeds with two sweeps over the scales in the compressed scale sequence. The first sweep is the \textit{cutting phase}, and the second the \textit{construction phase}.

    \mysubpara{Cutting phase.}
    The first sweep processes the scales~$\mathcal{R}$ from coarse to fine, and computes a collection of good cuts $\mathcal{C}_r$ for every scale $r$.
    In the cutting phase for scale $r$, we consider the union of cuts $\mathcal{C}_{\Lambda r}$ computed at the previous scale. Observe that, as we process the scales from coarse to fine, $\mathcal{C}_{\Lambda r}$ is empty whenever the previous scale is not $\Lambda r$.
    For each weakly simple polygon $P \in \supp{r}$, we then consider the subpolygons $\overline P$ that are created by cutting $P$ along the good cuts~$\mathcal{C}_{\Lambda r}$ inherited from scale~$\Lambda r$. Slightly abusing notation, we refer to these as the subpolygons in~$P\setminus \mathcal{C}_{\Lambda r}$.
    For every face $F$ in the junction tree $(\J_r,\F_r)$ for~$P$ we identify all $F$-spines $\spineset_F$ that are in~$\overline{P}$. Via \Cref{lem:construct_min_realization}, we then compute a good cut of the spine-cutting graph $G_{F,\spineset_F}$ or certify that $\opt\geq \frac{1}{100}\Lambda$, and output the MST.
    The union of all these good cuts makes up $\mathcal{C}_r$.

    \begin{algorithm}[!t]
\caption{An $\tilde{O}(\max\{n/\Lambda,\Lambda^{14}\})$-approximation algorithm for the minimum dilation tree.}
\label{alg:the_alg}
\begin{algorithmic}[1]
\Require Parameter $\Lambda\geq 6$, point set $S$
\Ensure Tree $T$ on $S$
\State Compute compressed scale sequence $\mathcal{R}$
\State Certify that either $S$ is $\Lambda$-thin, or $\opt\geq2\Lambda$ and output MST
\State Compute a consistent sequence of junction trees, or certify $\opt\geq\frac{1}{100}\Lambda$ and output MST
\For{scale $r$ in $\mathcal{R}$ in decreasing order}
    \Comment{\emph{cutting phase for $r$}}
    \State $\mathcal{C}_r \gets \emptyset$
    \State \textbf{if} $\mathcal{C}_{\Lambda r}$ is not defined \textbf{do} $\mathcal{C}_{\Lambda r}\gets\emptyset$
    \For{weakly simple polygon $P$ in $\supp{r}$}
        \State Let $(\J_r,\F_r)$ be the junction tree of $P$ in the consistent sequence
        \For{subpolygon $\overline P$ in $P \setminus \mathcal{C}_{\Lambda r}$}
            \For{face $F\in\mathcal{F}_r$}
                \State Let $\spineset_F$ be the subset of $F$-spines contained in $\overline P$
                \State Compute a good cut $C$ of $G_{F,\spineset_F}$, or certify that $\opt\geq\frac{1}{100}\Lambda$ and output MST
                \State $\mathcal{C}_{r} \gets \mathcal{C}_r \cup C$
            \EndFor
        \EndFor
    \EndFor
\EndFor
\State $T\gets(S,\emptyset)$
\For{scale $r$ in $\mathcal{R}$ in ascending order}
    \Comment{\emph{construction phase for $r$}}
    \For{weakly simple polygon $P$ in $\supp{r}$}
        \State Let $(\J_r,\F_r)$ be the junction tree of $P$ in the consistent sequence
        \For{subpolygon $\overline P$ in $P \setminus \mathcal{C}_{\Lambda r}$}
            \For{junction $J$ in $\J_r$}
                \State $S_J \gets \emptyset$ and let $\spineset_P$ be the subset of $P$-spines contained in $\overline P$
                \If{ $\spineset_P \cap J \neq \emptyset$ }
                    \State Let $P_J$ be all subpolygons in $\spineset_P\setminus\mathcal{C}_r$ that intersect $J$
                    \State Add to $S_J$, for each subpolygon $\overline \spine$ in $P_J$, a point in $S\cap \overline \spine$ closest to~$J$\footref{fn:closest}
                \ElsIf
                        {at least two adjacent faces intersect $S \cap \overline{P}$}
                    \State Add to $S_J$, for each face $F$ adjacent to $J$, a point in $S \cap F \cap \overline P$ closest to~$J$\footref{fn:closest}
                \EndIf
                \State Add to $T$ a star connecting the points in $S_J$
            \EndFor
            \For{face $F$ in $\F_r$}
                \State Add to $T$ a star connecting all connected components of $T$ that have a point \hspace*{2.2cm} in $(S \cap F \cap \overline P) \cup \bigcup_{J \in \junc(F)}S_J $ (using only points in $(S \cap F \cap \overline P) \cup \bigcup_{J \in \junc(F)}S_J) $\vspace{-0.15cm}
            \EndFor
        \EndFor
    \EndFor
\EndFor
\State \textbf{return} $T$
\end{algorithmic}
    \end{algorithm}

    \mysubpara{Construction phase.}
    The second sweep processes the scales from fine to coarse and constructs the tree $T$ in a single pass.
    Initially, the edge set of the tree $T$ is empty.
    At each scale~$r\in \mathcal{R }$, we process all subpolygons $\overline P$ in $P\setminus\mathcal{C}_{\Lambda r}$ for every weakly simple polygon $P\in\supp{r}$ separately. For each such $\overline P$ we add a set of edges to $T$, illustrated in \Cref{fig:star_connections}.
    The input to the construction phase for scale $r$ and a subpolygon $\overline P$ is:
    \begin{compactenum}
        \item The junction tree $(\J_r,\mathcal{F}_r)$ of $P$.
        \item For each face $F \in \F_r$, a good cut $C$ of the spine-cutting graph $G_{F,\spineset_F}$, where $\spineset_F$ is the set of $F$-spines contained in $\overline P$.
        \item A forest $T$ which, for every subpolygon $\overline \Sigma$ in  $\spineset_P \setminus \mathcal{C}_r$, contains a tree spanning $S \cap \overline{\spine}$, where $\spineset_P$ is the set of $P$-spines contained in $\overline P$.
    \end{compactenum}
    The construction phase for scale $r$ and subpolygon $\overline P$ merges these trees, so that the forest $T$ contains a tree that spans $S\cap \overline P$.

    \begin{figure}
        \centering
        \includegraphics{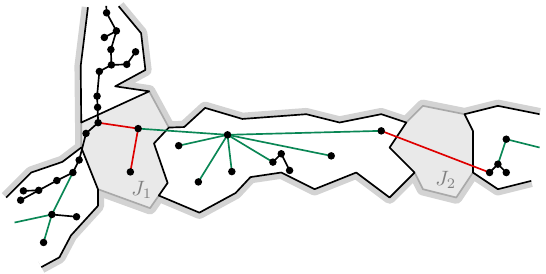}
        \caption{The star connections added in the construction phase of the algorithm to connect the points in $S \cap \overline P$. The black edges are in the forest $T$ constructed before the construction phase for~$r$. In the construction phase for $r$, we add the red edges when considering junction $J_1$ and $J_2$, and add the green edges when we consider each of the faces adjacent to these junctions.}
        \label{fig:star_connections}
    \end{figure}

    We first process each junction $J$ in $\J_r$ connecting points in (or close to) the junction (the red edges in~\Cref{fig:star_connections}).
    We consider two different cases based on $\spineset_P$.
    \refstepcounter{footnote}
    \footnotetext[\value{footnote}]{\label{fn:closest}Closest is measured using the geodesic distance in $P$.}

    \begin{itemize}[noitemsep,nolistsep]
    \item If $\spineset_P$ intersects $J$ (see $J_1$ in \Cref{fig:star_connections}), we add to $T$ a star that connects the set $S_J$, which contains one point from each subpolygon $\overline\Sigma$ in $\spineset_P\setminus\mathcal{C}_r$ that intersects $J$, where this point in $\overline\Sigma$ is any point in $J$, or, if such a point does not exist, the (geodesically in $P$) closest to $J$. 

    \item If $\spineset_P$ does not intersect $J$ (see $J_2$ in \Cref{fig:star_connections}), but two or more adjacent faces contain points of $S\cap\overline P$, then we add to $T$ a star connecting the set $S_J$, which contains a single point in $S$ from every face adjacent to the junction $J$ with $S\cap F \cap \overline P \neq \emptyset$, where this point is the (geodesically in~$P$) closest to $J$.
    \end{itemize}
    Having processed for $\overline P$ all junctions in $\J_r$, we process each face in $\F_r$.
    For a face $F$, we consider all connected components of~$T$ whose vertices are contained in ${\overline P}\cap F$ or in the set $S_J$ for some junction $J \in \junc(F)$, and add a star to $T$ connecting all of these components (the green edges in \Cref{fig:star_connections}). We ensure the edges of this star only connect points of $S$ in ${\overline P} \cap F$ or $S_J$ for some~$J \in \junc(F)$.

    \subsection{Correctness}\label{sec:correctness_algorithm}

    We prove correctness of the algorithm using the following four invariants that hold at the end of each construction phase, under the condition that $\Lambda \geq 6$.
    After the construction phase for $r$, we have for each weakly simple polygon $P \in \supp{r}$ that:
    \smallskip
    \begin{description}[nosep]
    \item[Connectedness]\phantomsection\label{inv:conn} The graph $T$ consists of a tree for each subpolygon $\overline P$ in $P \setminus\mathcal{C}_{\Lambda r}$ that connects exactly $S\cap \overline{P}$.
    \item[Packedness]\phantomsection\label{inv:pack} Edges added in phase $r$ are either
    (i) junction-star edges of length at most $24\Lambda r$, or
    (ii) face-star edges of length at most $61(\Lambda^2\log n)r$.
    Furthermore, we can associate each face-star edge with a unique $P$-spine that contains one of its endpoints.
    \item[Locality]\phantomsection\label{inv:loc} For any two points $p,q$ in the same subpolygon $\overline P$ in $P\setminus\mathcal{C}_{\Lambda r}$ for which $\dist_P(p,q) < \frac{2}{3}\Lambda^2 r$, all points of $S$ in the shortest path $\pi_T(p,q)$ are contained in the geodesic disk $D_{P}(p,(200\log n + 10)\Lambda^2r)$.
    \item[Dilation]\phantomsection\label{inv:dil}
    For every subpolygon $\overline P$ of $P\setminus\mathcal{C}_{\Lambda r}$, the dilation in $T$ of any two points $p,q\in S\cap\overline P$ with $\dist_P(p,q)\leq 2r$ is at most $29\,965\,648\,003\,370\,\log^3 n \max(n/\Lambda,\Lambda^{14})$.
    \end{description}
    \smallskip
    Recall that we assume that $\Lambda \geq 6$. Next, we prove that each of the four invariants hold throughout the algorithm. The invariants then directly imply that when the algorithm terminates, the produced graph $T$ is a tree on $S$ with bounded dilation.

    \begin{lemma}\label[lemma]{lem:well-spaced-correctness}
        Suppose the algorithm does not terminate before the construction phase. Then for every $r\in\mathcal{R}$ and every weakly simple polygon $P\in\supp{r}$, the good cuts in $\mathcal{C}_r$ contained in $P$ are well-spaced w.r.t.\ the set of all $P$-spines.
    \end{lemma}
    \begin{proof}
        In the cutting phase for $r$, the good cuts in $\mathcal{C}_r$ contained in $P$ are computed via \cref{lem:construct_min_realization} for each subpolygon $\overline P$ in $P\setminus\mathcal{C}_{\Lambda r}$ and each face $F\in\F_r$, for the spine-cutting graph $G_{F,\spineset_F}$, where~$\spineset_F$ is the set of $F$-spines induced by the $P$-spines contained in $\overline P$. As $\mathcal{C}_{\Lambda r}$ does not intersect $\supp{\Lambda^{-1}r}$ by property~\ref{itm:not_intersect_nerve} of~\cref{def:minimal_realization}, every $P$-spine is contained in exactly one subpolygon $\overline P$. Hence, the sets of $P$-spines contained in the subpolygons $\overline P$ are pairwise disjoint and their union is the set of all $P$-spines, and the claim follows from \cref{cor:well-spaced-cuts}.
    \end{proof}

    \begin{lemma}[{\hyperref[inv:conn]{Connectedness}} invariant]
   At the end of the construction phase for $r$, the graph $T$ is a forest consisting of exactly one tree for every subpolygon $\overline P$ in $P\setminus \mathcal{C}_{\Lambda r}$ for every weakly simple polygon $P$ in $\supp{r}$, which spans the points of $S$ contained in $\overline P$.
    \end{lemma}
    \begin{proof}
        We prove the lemma by induction on $r$. We start with $T = (S,\emptyset)$ where each point is its own degenerate polygon in the support as base case, where the statement holds trivially.
        Now, suppose the invariant holds for all scales~$r'\in\mathcal{R}$ smaller than~$r$.

        Observe that when $\Lambda^{-1}r \notin \mathcal{R}$, then for the scale $s$ preceding $r$ we have that $\mathcal{P}_s = \mathcal{P}_{\Lambda^{-2}r} = \mathcal{P}_{\Lambda^{-1}r}$, because by definition of the compressed scale sequence there are no point-pairs with pairwise distance in $[s, r)$. So, every $F$-spine at scale $r$ is also an $F$-nerve. By property~\ref{itm:not_intersect_nerve} of a good cut, this forces $\mathcal{C}_r = \emptyset$, ensuring $\mathfrak{S}_P \setminus \mathcal{C}_r = \mathcal{P}_s \cap \overline{P}$. For the rest of the proof, we assume that the scale preceding~$r$ is $\Lambda^{-1}r$.

        First, observe that in the construction phase for $r$, we consider for each $P \in \supp{r}$ each subpolygon~$\overline P$ in $P \setminus \mathcal{C}_{\Lambda r}$ separately.
        Therefore, no edge between any pair of points $p,q$
        in different subpolygons of $\supp{r} \setminus \mathcal{C}_{\Lambda r}$ is added to $T$.
        Furthermore, the good cuts in~$\mathcal{C}_{\Lambda r}$ do not intersect $\supp{\Lambda^{-1} r}$ by property~\ref{itm:not_intersect_nerve} of~\Cref{def:minimal_realization}, so they also do not intersect~$\supp{\Lambda^{-2}r}$.
        Hence there are also no edges from any smaller scale $r' < r$ between points in different subpolygons of~$\supp{r}\setminus \mathcal{C}_{\Lambda r}$.

        Next, consider the graph $G$ that is the induced subgraph of the graph $T$ at the end of the construction phase for scale $r$, on the points of $S$ contained in a subpolygon~$\overline P$ in~$P\setminus \mathcal{C}_{\Lambda r}$ for some weakly simple polygon $P \in \supp{r}$. Let $\spineset_P$ be the set of $P$-spines contained in $\overline{P}$.

        \mysubpara{\boldmath $G$ is cycle-free.} We first show that $G$ does not contain any cycles. Suppose for contradiction that $G$ is not a forest, and thus contains a cycle. By induction, at the start of phase $r$, we have a forest in every subpolygon $\overline \Sigma$ of $\spineset_P\setminus \mathcal{C}_r$.
        Hence, any cycle must come from the addition of edges in the construction phase for~$r$.
        When we add edges for each face in $\F_r$, we only add a star between distinct connected components in~$T$, i.e., we consider \emph{the entire} forest constructed so far. We thus do not create a cycle due to any of these stars.
        It follows that any cycle must have been created when we added stars for each of the junctions.

        First note that the points in the subset $S_J \subset S$ between which we add a star for a junction~$J$ must all be in different connected components of the tree $T$ at the beginning of this phase:
        all points we add to~$S_J$ are in different subpolygons of $\spineset_P \setminus \mathcal{C}_r$ (if $\spineset_P\cap J=\emptyset$, they lie in different faces adjacent to $J$, and no subpolygon crosses $J$) and are thus in different connected components by the induction hypothesis.
        Hence, there must be a minimal cycle $W$ that uses star edges of at least two different junctions in $\J_r$, see \Cref{fig:no_cycles_proof}. Let $J_1,J_2,\ldots,J_k$ be the junctions whose star edges are used by this cycle $W$, and let $W_1,W_2,\ldots,W_k$ be the subpaths of $W$ from each star edge to the next, i.e., $W_1$ connects a point (close to) $J_1$ with a point (close to) $J_2$. Note that the edges in the paths $W_i$ all come from earlier phases. Let $Q_1,Q_2,\ldots,Q_k$ be the subpolygons of $\spineset_P \setminus \mathcal{C}_r$ that contain the points of $S$ in $W_1,\ldots,W_k$, respectively.

        \begin{figure}
            \centering
            \includegraphics{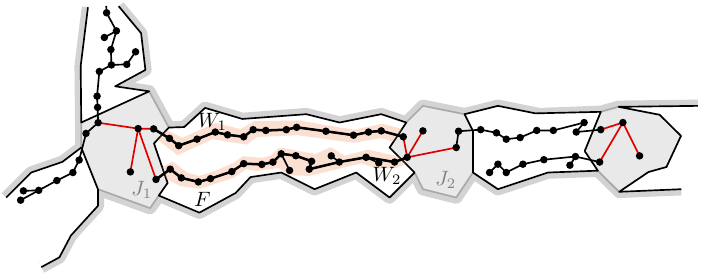}
            \caption{
                A cycle in $T$ created at scale $r$ must be due to the added junction-stars.
                However, this would imply that one of the indicated $F$-spines would have been {cut by $\mathcal{C}_{ r}$}.
            }
            \label{fig:no_cycles_proof}
        \end{figure}

        By the properties of a star, each star is traversed by $W$ at most once.
        Further, a star from a junction $J$ for which $\spineset_P \cap J = \emptyset$ cannot participate in this cycle, as this star is the only way to traverse from one of its neighbouring faces into another: any other path from an earlier phase crossing the junction lies in a single subpolygon of $\spineset_P\setminus\mathcal{C}_r$, which would imply $\spineset_P \cap J \neq \emptyset$.
        We show that this implies that there is a face $F \in \F_r$ such that there is a relevant vertex of $G_{F,\spineset_F}$ whose corresponding region is enclosed by $F$-spines in $\spineset_F$ that are not cut by $\mathcal{C}_r$, where $\spineset_F$ is the set of $F$-spines induced by $\spineset_P$.

        Let $F$ be any face in $\mathcal{F}_r$ whose removal disconnects some of the junctions $J_1,\ldots,J_k$ from the others. Recall that the star added for junction $J$ with $\spineset_P \cap J \neq \emptyset$ connects subpolygons of $\spineset_P \setminus \mathcal{C}_r$ that intersect the junction.
        So, each $Q_i$ that crosses $F$ (and so $S \cap W_i \cap F \neq \emptyset$)
        must also intersect two junctions in $\junc(F)$. Furthermore, as the subpolygons of $\spineset_P \setminus \mathcal{C}_r$ that are connected by a junction-star are distinct (and by definition disjoint), each such $Q_i$ gives rise to at least one $F$-spine connecting two junctions. As $W$ is a cycle that visits at least two
        distinct junctions that neighbour~$F$ in the junction tree,
        these $F$-spines, together with $\junc(F)$, enclose one or more polygons of $F\setminus \spineset_F$ that intersect at least two junctions.
        The cycle formed by the paths $Q_i\cap F$ and $\junc(F)$ is disjoint from~$\mathcal{C}_r$ and from all vertex regions of $G_{F,\spineset_F}$. The region it encloses either contains a region of~$F\setminus\spineset_F$ touching two junctions (a relevant vertex), or, if two of the $Q_i$ lie in the same $F$-spine, the edges of $C$ that separate them. In either case, the enclosed relevant vertex or cut edge cannot lie in the connected component of $C\cup\bigcup(\text{vertex regions})$ containing the outer face, since reaching it would require crossing one of the $Q_i$. This contradicts property~\ref{itm:tree} of good cuts.
        We conclude that~$G$ does not contain any cycles.

        \mysubpara{\boldmath $G$ is connected.} What remains is to show that $G$ is connected, i.e., $G$ is a tree. We prove that any pair of points $p,q$ that are vertices of $G$ with $\dist_P(p,q) \leq 3\Lambda r$ is connected in $G$. Because $\mathcal{C}_{\Lambda r}$ is \hyperref[def:well-spaced-cut]{well-spaced} by \cref{lem:well-spaced-correctness} (or by \cref{lem:well-spaced-cuts} applied to $\mathcal{C}=\emptyset$ if $\mathcal{C}_{\Lambda r}=\emptyset$), it then follows that $G$ is connected.

        Because we have $\dist_P(p,q) \leq 3\Lambda r$, and the geodesic distance between two junctions in~$\J_r$ is at least~${\frac{2}{3}\Lambda^2 r \geq 3\Lambda r}$, the points~$p,q$ must either be: (1) in the same face, (2) in the same junction, (3) in a face and an adjacent junction, or (4) in two faces that are separated by a single junction. We consider each of the four cases.
        \begin{enumerate}[nolistsep,noitemsep]
        \item If $p,q$ are in the same face $F$, the star for the face $F$ connects the components of $T$ containing $p$ and $q$.
        \item If $p,q$ are in the same junction $J$, then the star for junction $J$ will connect the components of $T$ containing $p$ and $q$.
        \item If $p$ is contained in a face $F$ and $q$ in a junction $J \in \junc(F)$, then the face-star for $F$ will connect the components of $T$ containing $p$ and $q$ if they have not been connected by the star for junction $J$.
        \item If $p$ and $q$ are in faces $F_1,F_2$ that are separated by a single junction $J$, then the stars for the faces $F_1$ and $F_2$ both connect to points in $S_J$, which are connected by the star of junction $J$.
        \end{enumerate}
        We conclude that $p$ and $q$ are connected in $G$, and thus $G$ is a tree.
    \end{proof}

    \begin{lemma}[{\hyperref[inv:pack]{Packedness}}  invariant]
    The edges added in the construction phase for $r$ are either
    junction-star edges of length at most $24\Lambda r$, or
    face-star edges of length at most~$(61\Lambda^2\log n)r$.
    Furthermore, we can associate each face-star edge with a unique $P$-spine of some weakly simple polygon $P$ in $\supp{r}$ that contains one of its endpoints.
    \label[lemma]{lem:packedness}
    \end{lemma}
    \begin{proof}
        Let $P$ be a weakly simple polygon in $\supp{r}$ and $\overline P$ a subpolygon of $P\setminus \mathcal{C}_{\Lambda r}$. Let $\spineset_P$ be the set of $P$-spines contained in $\overline{P}$, and for a face $F\in\F_r$ let $\spineset_F$ be the set of $F$-spines induced by $\spineset_P$. First, consider a junction $J\in \J_r$ with the corresponding set $S_J$.
        All edges in the star for~$J$ connect pairs of points in $S_J$.

        We argue that all points in $S_J$ have geodesic distance (in $P$) at most $3\Lambda r$ to $J$, whenever $S_J \neq \emptyset$.
        In the case $\spineset_P \cap J \neq \emptyset$, each point in $S_J$ is in some subpolygon of $\spineset_P \setminus \mathcal{C}_r$ that intersects $J$. In particular, the boundary of this subpolygon must intersect $J$. Because $\mathcal{C}_r$ is disjoint from the $z_2$-buffer of $J$ and the maximum edge length of the boundary of $\spineset_P$ is $r/\Lambda < z_2 = r/2$, there is a vertex of the boundary of such a subpolygon that has distance at most $3r$ to $J$.
        In the case~$\spineset_P \cap J = \emptyset$, the tree structure of the junction tree $(\J_r,\F_r)$ together with the well-spacedness of $\mathcal{C}_{\Lambda r}$ (\cref{lem:well-spaced-correctness}) implies that any junction $J \in \J_r$ for which $|S_J| \geq 2$ has, for any adjacent face with a point of $S\cap \overline P$ in it, also a point of geodesic distance at most $3\Lambda r$ to $J$. As the geodesic diameter of the junction is bounded by~$18 \Lambda r$ due to~\cref{def:the-junction-tree}, the pairwise distance between any two points in $S_J$ is at most~$18\Lambda r + 2\cdot(3\Lambda r) = 24 \Lambda r$.

        Next, we bound the length of edges in face-stars.
        Consider a face $F\in\F_r$, and let $J$ be an adjacent junction.
        The geodesic diameter of~$F$ is at most~$(54\Lambda^2\log n)r$, so the Euclidean distance from any point in $F$ to a point in $S_J$ is at most~$(54\Lambda^2\log  n)r + 3 \Lambda r + 18 \Lambda r$.
        Furthermore, for any second junction $J'$ also adjacent to $F$, the Euclidean distance between two points of~$S_J$ and $S_{J'}$ is bounded by
        $2\cdot (18 \Lambda r + 3 \Lambda r) + (54\Lambda^2\log  n)r$, which is at most $(61\Lambda^2\log n)r$ for~{$\Lambda \geq 6$}.
        It remains to show that these edges connect points in different $P$-spines.

        Consider two points $p,q$ of $S$ that lie in the same $P$-spine $\spine\in\spineset_P$. If $\spine$ is fully contained in a face $F$, then property~\ref{itm:in_spines} of a good cut ensures that~$p$ and $q$ are in the same subpolygon of~${\spineset_P \setminus \mathcal{C}_r}$, and thus the \hyperref[inv:conn]{Connectedness} invariant implies that $p$ and $q$ were already connected in $T$ at the start of phase $r$. If $\spine$ intersects a junction, then a path from $p$ to $q$ in $\spine$ alternates between faces and junctions, and within each such face~$F$ it can be shortcut to a path in $\spineset_F  \cup \junc(F)$. But then property~\ref{itm:minimal} and property~\ref{itm:path_junctions} of a good cut imply that there is a path in $(\spineset_F  \cup \junc(F)) \setminus \mathcal{C}_r$ connecting $p$ and $q$ through these faces and junctions. It follows that the junction-stars together with the trees spanning each $\overline\Sigma$ in $\spineset_P\setminus\mathcal{C}_r$ connect $p$ and $q$. Hence, before we insert a face-star, any two points in the same $P$-spine in $\spineset_P$ already belong to the same connected component of~$T$.
        So, any face-star edge connects two points from different $P$-spines.
        In fact, each such edge can be associated with one of its endpoints such that any pair of associated points comes from different $P$-spines. For this, root the trees of~$T$ arbitrarily, and associate to each star-edge the point further away from the root.
        This even holds across different subpolygons in~${\supp{r}\setminus\mathcal{C}_{\Lambda r}}$,
        as~$\mathcal{C}_{\Lambda r}$ does not intersect $\supp{\Lambda^{-1}r}$ by definition.
    \end{proof}

    \begin{lemma}[{\hyperref[inv:loc]{Locality}} invariant]
    Let $T$ be the graph at the end of the construction phase for~$r$. For any two points $p,q \in S$ in the same subpolygon $\overline P$ in $P\setminus\mathcal{C}_{\Lambda r}$ for some weakly simple polygon~$P$ in~$\supp{r}$ with $\dist_P(p,q) < \frac{2}{3}\Lambda^2r$, all points of $S$ in the shortest path $\pi_T(p,q)$ between $p$ and $q$ in $T$ are contained in the geodesic disk $D_{P}(p,(200\log n + 10)\Lambda^2r)$.
    \end{lemma}

    \begin{proof}
        Let $\delta = (200\log n+10)$. We prove the lemma by induction on $r$. Since $r_0<\|p-q\|$ for any~${p,q\in S}$, we have that $\supp{r_0}$ consists of only singleton points and the base case thus holds trivially. Now, suppose the invariant holds for all scales~$r'$ smaller than~$r$.

        Let $P$ be a weakly simple polygon in $\supp{r}$ with junction tree $(\J_r,\F_r)$ and let $\overline P$ be a subpolygon in $P \setminus \mathcal{C}_{\Lambda r}$. Let $\spineset_P$ be the set of $P$-spines contained in $\overline P$, and for a face $F\in\F_r$ let $\spineset_F$ be the set of $F$-spines from $\spineset_P \cap F$. Recall that $\mathcal{C}_{\Lambda r}$ does not intersect $\supp{\Lambda^{-1}r}$ by property~\ref{itm:not_intersect_nerve} of a good cut, so every $P$-spine is either contained in $\overline P$ or disjoint from it, and every point of $S\cap\overline P$ lies in a unique subpolygon of $\spineset_P\setminus\mathcal{C}_r$. By the \hyperref[inv:conn]{Connectedness} invariant, $T$ is a forest and, at the start of the construction phase for $r$, contains a tree spanning $S\cap\overline{\spine}$ for every subpolygon $\overline{\spine}$ in~$\spineset_P\setminus\mathcal{C}_r$.
        Throughout, let $R = 162(\Lambda^2\log n)r+54\Lambda r+3r+\delta\Lambda r$.

        We first prove two claims on the structure of shortest paths in $T$.

        \begin{claim}\label[claim]{claim:path-in-subpolygon}
            Let $\overline\spine$ be a subpolygon in $\spineset_P\setminus\mathcal{C}_r$ and $X$ be a face or junction in $\F_r \cup \J_r$. Then for any $a,b\in S\cap\overline\spine$, where there is a path from $a$ to $X$ that crosses at most one other face or junction (not containing $a$ or $x$), and similarly for $b$, all points of $S$ on $\pi_T(a,b)$ lie in $D_P(x,R)$ for every $x\in X$.
        \end{claim}
        \begin{claimproof}
            Let $\gamma$ be a path in $\overline\spine$ from $a$ to $b$ that visits a minimum number of distinct junctions among all such paths. We first argue that $\gamma$ never crosses a face~$F'$ from a junction~$J$ to another junction in~$\junc(F')$ and later back into $J$.
            If $\gamma$ did cross a face twice in this way, then~$\gamma$ would separate a region of $F'$ from the region of the outer-face vertex in $G_{F',\spineset_{F'}}$, as illustrated in~\cref{fig:cases}(a). By property~\ref{itm:tree} of the good cut $C\in\mathcal{C}_r$ of $G_{F',\spineset_{F'}}$, this region contains neither a relevant vertex of $G_{F',\spineset_{F'}}$, nor an edge of $C$, as $\gamma$ and the junctions do not intersect~$C$ (c.f.\ the proof of the {\hyperref[inv:conn]{Connectedness}} invariant). So, there must be a part of $\overline \Sigma$ that connects the two subpaths of $\gamma$ crossing $F'$, and thereby $\gamma$ can be shortcut such that it visits at least one fewer junction, contradicting its choice.

            By the restriction of $a$ and $b$ that their path to $X$ crosses at most one other junction or face, we conclude that $\gamma$ is contained in~${D_P(x, 3 \cdot (54 \Lambda^2 \log n)r + 3 \cdot 18\Lambda r)}$ for any point $x \in X$.
            By \cref{cor:well-spaced-cuts}, there is a walk from~$a~$ to~$b$ along Delaunay edges of length at most $3r<\frac23\Lambda r$ in $\overline\spine$ whose vertices lie within geodesic distance~$3r$ of $\gamma$.
            So, applying the induction hypothesis on each consecutive pair of points on this path gives that all points of $S$ on $\pi_T(a,b)$ lie in $D_P(\gamma,3r+\delta\Lambda r) \subseteq D_P(x,3 \cdot (54 \Lambda^2 \log n)r + 3 \cdot 18\Lambda r + 3r+\delta\Lambda r) = D_P(x,R)$ for any $x \in X$.  (If the scale $\Lambda^{-1}r$ does not appear in the compressed scale sequence $\mathcal{R}$, we apply the induction hypothesis to the preceding scale $s$ instead, using that~$\supp{\Lambda^{-1}r}=\supp{s}$, $\mathcal{C}_{\Lambda s}=\emptyset$, and the edges have length at most~$3\Lambda s<\frac23\Lambda^2 s$ by \cref{lem:max_edge_length_support}.)
        \end{claimproof}

        \begin{figure}[!t]
            \centering
            \includegraphics{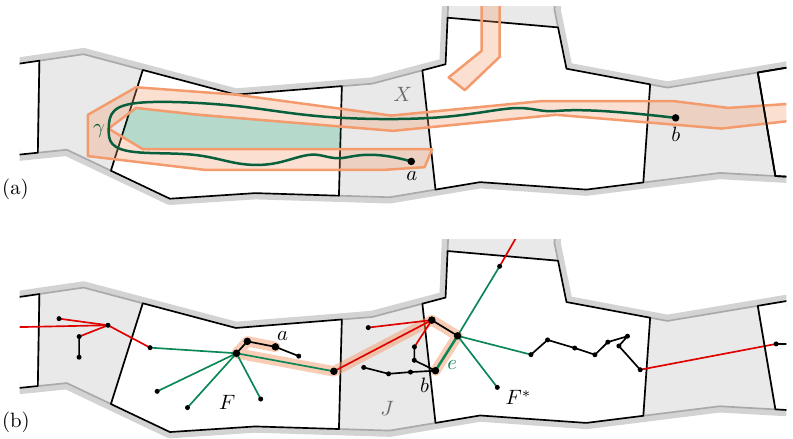}
            \caption{
                (a) The path $\gamma$ connecting $a$ and $b$ in $\overline\spine$ in the proof of \cref{claim:path-in-subpolygon}.
                (b) A face-star edge from a neighbouring face on the shortest path between $a$ and $b$ induces a cycle with the junction-star of the junction we pass through in the proof of \cref{claim:path-near-face}.
            }
            \label{fig:cases}
        \end{figure}

        \begin{claim}\label[claim]{claim:path-near-face}
            Let $F\in\F_r$ and let $a,b\in S\cap\overline P$, with $a \in F$ and either $b \in F$, $b \in J$ for some $J\in\junc(F)$, or $b\in S_J$ for some $J\in\junc(F)$. Then
            all points of $S$ on $\pi_T(a,b)$ lie in $D_P(a,R)$.
        \end{claim}
        \begin{claimproof}
            For a junction $J'\in\junc(F)$, we say that a vertex of $T$ is \emph{attached} to $J'$ if it lies in a subpolygon of $\spineset_P\setminus\mathcal{C}_r$ that contains a point of $S_{J'}$. In particular, all points in $J'$ and $S_{J'}$ are attached to $J'$. For any edge of $T$ whose endpoints do not both lie on the same side of~$J'$ (one of them lies in $J'$, or they lie on different sides), its endpoints are inside a subpolygon of~$\spineset_P\setminus\mathcal{C}_r$ intersecting $J'$, and thus it is an edge of a previous scale, or an edge of the junction-star of $J'$, or an edge of a face-star with an endpoint in $S_{J'}$. In other words, it has an endpoint attached to $J'$.

            Next, we show that every star edge added in phase $r$ that is on $\pi_T(a,b)$ belongs to the face-star of $F$ or to the junction-star of a junction in $\junc(F)$.
            Suppose, for contradiction, that~$\pi_T(a,b)$ contains a star edge $e=uv$, with $u$ closer to $a$ along $\pi_T(a,b)$, that belongs neither to the face-star of $F$ nor to the junction-star of a junction in $\junc(F)$, see \cref{fig:cases}(b). We distinguish two cases.

            \emph{The edge $e$ belongs to the face-star of a face~$F^*\neq F$.} Let $J'\in\junc(F)$ be the junction separating~$F$ from~$F^*$. Then~$u$ lies in~$F^*$ or in~$S_{J''}$ for some $J''\in\junc(F^*)$. It follows that $u$ is attached to~$J'$ (if $J''=J'$) or lies on the other side of~$J'$ than $F$.

            \emph{The edge $e$ belongs to the junction-star of a junction $J^*\notin\junc(F)$.} Let $J'\in\junc(F)$ be a junction separating $F$ from $J^*$. As $u\in S_{J^*}$, $u$ lies in either $J^*$ or one of the faces neighbouring $J^*$.
            So $u$ must lie on the other side of $J'$ than $F$.

            In both cases, there is an edge of~$\pi_T(a,u)$ whose endpoints do not both lie on the same side of~$J'$, and thus~$\pi_T(a,u)$ contains a vertex~$u'$ attached to~$J'$. Likewise, $\pi_T(v,b)$ contains a vertex~$v'$ attached to $J'$: either $v$ or $b$ is attached, or $b$ lies on the side of $F$ of $J'$ (as $b$ lies in $F$, in a junction~${J''\neq J'}$ of~$F$, or in $S_{J''}$) and $v$ on the other side. The junction-star of $J'$ together with the trees of the subpolygons of $\spineset_P\setminus\mathcal{C}_r$ containing $u'$ and $v'$ connects $u'$ and $v'$ without using $e$. We thereby find a cycle in $T$,  contradicting the \hyperref[inv:conn]{Connectedness} invariant.

            It follows that we can split $\pi_T(a,b)$ into edges of the face-star of $F$, edges of junction-stars of junctions in $\junc(F)$, and subpaths that are inside a subpolygon of $\spineset_P\setminus \mathcal{C}_r$. All endpoints of these subpaths lie in $F$, in a junction in $\junc(F)$, or in a face adjacent to a junction of $F$. We can thus apply \cref{claim:path-in-subpolygon} with $X=F$ and $x=a$ on each subpath to conclude that all points of $S$ on~$\pi_T(a,b)$ lie in $D_P(a,R)$.
        \end{claimproof}
        Let $p,q \in S$ be two points in $\overline P$ with $\dist_P(p,q) < \frac{2}{3}\Lambda^2r$. If $p$ and $q$ lie in the same face $F$, then \cref{claim:path-near-face} immediately implies the claim. Otherwise, the shortest path in $P$ from $p$ to $q$ intersects at most one junction by the lower bound on the pairwise geodesic distance of junctions. So there is a junction $J$ such that $p$ and $q$ lie in $J$ or in a face adjacent to $J$.

        We define intermediate points $s_p,s_q \in S_J$ as follows.
        If $\spineset_P \cap J \neq \emptyset$, we have $S_J \neq \emptyset$.
        Otherwise, neither $p$ nor $q$ is in $J$.
        Having previously handled the same-face case (\Cref{claim:path-in-subpolygon}), we consider the case that they lie in distinct faces adjacent to $J$, both of which contain a point in~$S\cap\overline{P}$.
        Hence, we have that~$S_J \neq \emptyset$.
        If $p\in J$, let $s_p$ be a point of $S_J$ in the subpolygon of~$\spineset_P\setminus\mathcal{C}_r$ containing $p$. As this subpolygon intersects $J$, there is indeed such a point in $S_J$. Otherwise, let $s_p$ be an arbitrary point in $S_J$. Choose $s_q$ in the same way. Because $T$ is a forest, we have~${\pi_T(p,q)\subseteq\pi_T(p,s_p)\cup\pi_T(s_p,s_q)\cup\pi_T(s_q,q)}$. As $s_p$ and $s_q$ are in $S_J$, $\pi_T(s_p,s_q)$ consists of at most two edges of the junction-star of $J$.

        If $p\in J$, then $p$ and $s_p$ lie in the same subpolygon of $\spineset_P\setminus\mathcal{C}_r$, we can apply \cref{claim:path-in-subpolygon} with~${X=J}$ and~${x=p}$ to find~$\pi_T(p,s_p)\subseteq D_P(p,R)$. If~$p \notin J$, then~$p$ lies in a face $F_p$ adjacent to $J$, so \cref{claim:path-near-face} with $F=F_p$, $a=p$ and $b=s_p$ again gives $\pi_T(p,s_p)\subseteq D_P(p,R)$. In the same way, we get~${\pi_T(s_q,q)\subseteq D_P(q,R)\subseteq D_P(p,R+\frac23\Lambda^2 r)}$, as $\dist_P(p,q)<\frac23\Lambda^2 r$. Using that $R+\frac23\Lambda^2 r\leq\delta\Lambda^2 r$ for~$\Lambda\geq 6$, we conclude that all points of $S$ on~$\pi_T(p,q)$ are contained in~$D_P(p,\delta \Lambda^2 r)$.
    \end{proof}

    To prove the dilation invariant, we use the following geometric packing observations.

    \begin{observation}\label[observation]{obs:pack}
        Let $D(p,R)$ be a Euclidean disk, and let $X\subset D(p,R)$ be a set of at least two points where every pair of points in $X$ has distance at least $r$. Then $|X|\leq16\frac{R^2}{r^2}$.
    \end{observation}
    \begin{proof}
        Trivially, $r\leq 2R$. By centring disks of radius $r/2$ at every point of $X$, we obtain $|X|$ disjoint disks of radius $r/2$, that are all contained in $D(p,2R)$. As $D(p,2R)$ has area $4\pi R^2$, and the disjoint disks have total area $|X|\frac{\pi r^2}{4}$, we obtain~$|X|\leq 16\frac{R^2}{r^2}$.
    \end{proof}

    \begin{observation}
        Any two points $p,q\in S$ in different connected components of the $r$-support have Euclidean distance at least $r$.
        \label[observation]{obs:pairwise-component-distance}
    \end{observation}
    \begin{proof}
        Recall that the Euclidean MST $M$ is a subgraph of the Delaunay triangulation.
        Select any two points with $\lVert p-q\rVert < r$.
        If $pq$ is not an edge of the MST, the path $\pi_{M}(p,q)$ uses only edges of length at most $\lVert p-q\rVert$, as otherwise we could obtain a cheaper tree by removing an edge and inserting $pq$.
        It follows that $\pi_{M}(p,q)$ is contained in a connected component of $\supp{r}$.
    \end{proof}

    \begin{lemma}[{\hyperref[inv:dil]{Dilation}} invariant]\label[lemma]{lem:dilation}
        Let $T$ be the graph at the end of the construction phase for~$r$. For every subpolygon $\overline P$ of $P\setminus\mathcal{C}_{\Lambda r}$, where $P\in\supp{r}$, the dilation in $T$ of points $p,q\in S\cap\overline P$ with $\dist_P(p,q)\leq 2r$ is at most
        \[29\,965\,648\,003\,370\,\log^3 n \max(n/\Lambda,\Lambda^{14}).\]
    \end{lemma}
    \begin{proof}
        Let $\delta=(200\log n+10)$. We prove the lemma by induction on $r$. Since $r_0<\|p-q\|$ for any~${p,q\in S}$, we have that $\supp{r_0}$ consists of only singleton points and the base case thus holds trivially. Now, suppose the invariant holds for all scales~$r'$ smaller than~$r$. Let $p,q \in S$ that are in the same subpolygon $\overline P$ of $P \setminus \mathcal{C}_{\Lambda r}$ for a weakly simple polygon $P \in  \supp{r}$, with $\dist_P(p,q) \leq 2r$.

        Consider first the case where $\|p-q\|<r/(2\Lambda^2)$. Let $s$ be the scale preceding $r$ in the compressed scale sequence $\mathcal{R}$. We argue that $p$ and $q$ are in the same subpolygon of $\supp{s}\setminus \mathcal{C}_{\Lambda s}$ and that their geodesic distance in $\supp{s}$ is at most $2s$ using the path $\pi_\mathrm{Del}(p,q)$ between $p$ and $q$ in the Delaunay triangulation of $S$. As the Delaunay triangulation is a 2-spanner~\cite{DBLP:journals/siamcomp/Xia13}, $\pi_\mathrm{Del}(p,q)$ has length at most~$2\|p-q\|$.

        If $r=\Lambda s$, then every edge of $\pi_\mathrm{Del}(p,q)$ has length at most $s/\Lambda$, so the path is contained in a single weakly simple polygon in $\supp{\Lambda^{-1}s}$.
        Thus, $p$ and $q$ lie in the same weakly simple polygon of~$\supp{\Lambda^{-1}s}$ and their geodesic distance in $\supp{s}$ is at most $2s$. Further, because $\supp{\Lambda^{-1}s}$ is not intersected by $\mathcal{C}_r$ by property~\ref{itm:not_intersect_nerve} of a good cut, they are in the same subpolygon of $\supp{s}\setminus \mathcal{C}_{\Lambda s}$.

        If instead $r\neq \Lambda s$, then $r\geq\Lambda^{10}s$. Furthermore, by definition of a compressed scale sequence there is no pair of points in $S$ with pairwise distance in $[s,r)$. Hence $\|p-q\|<s$. Because there are no pairs in $S$ with pairwise distance between $s$ and $r$, every edge of $\pi_\mathrm{Del}(p,q)$ has length less than $s$ and is thereby contained in $\supp{s}$. So, $p$ and $q$ are in the same weakly simple polygon of $\supp{s}$ and their geodesic distance in $\supp{s}$ is at most $2s$. Finally, because $\mathcal{C}_{\Lambda s}$ is empty, they are in the same subpolygon of $\supp{s}\setminus\mathcal{C}_{\Lambda s}$.

        Because $p$ and $q$ are in the same subpolygon of $\supp{s}\setminus\mathcal{C}_{\Lambda s}$, and their geodesic distance in the corresponding polygon of $\supp{s}$ is at most $2\|p-q\|<2s$, we directly obtain the required bound by applying the induction hypothesis.

        Next, consider the case where $\|p-q\|\geq r/(2\Lambda^2)$. By the {\hyperref[inv:conn]{Connectedness}} invariant, $p$ and $q$ are connected in $T$. Since $\dist_P(p,q)\leq 2r<\frac{2}{3}\Lambda^2r$, the {\hyperref[inv:loc]{Locality}} invariant applies. Consequently, all points in $S$ on $\pi_T(p,q)$ are contained in $D=D_{P}(p,\delta\Lambda^2r)$. To bound the dilation, we bound the total number of junction- and face-stars of each scale up to $r$ that intersect $D$, by repeatedly applying \Cref{obs:pack}. As each edge of $\pi_T(p,q)$ comes from some junction- or face-star from a construction phase of scale $r'\leq r$, this bounds the length of $\pi_T(p,q)$.

        First, consider junction-stars added in the construction phase for scale $r$. There is at most one junction-star for every junction intersected by $\overline P$.
        No pair of points in $S_J \times S_{J'}$ with $J \neq J'$ can have Euclidean distance less than $r/2$: again the shortest path in the Delaunay triangulation gives a path in $P$ of length $r\leq \frac{2}{3}\Lambda^2 r - 6 \Lambda r$, contradicting the properties of a junction tree.
        Thus, we can select a point in $S_J$ for each junction $J$ and apply \Cref{obs:pack} to bound the number of junction-stars in $D$ that were added in the construction phase for $r$ by $16\frac{(\delta\Lambda^2r)^2}{r^2/4}=64\delta^2\Lambda^4$. The maximum length of such an edge is~$24\Lambda r$ by the {\hyperref[inv:pack]{Packedness}} invariant.

        Next we consider junction-stars from all scales $r'$ with $r'<r$. Let us first consider the junction-stars added in the construction phase of the previous scale~$\Lambda^{-1}r$.
        Here, we have to account for every junction-star of every subpolygon of $\supp{\Lambda^{-1}r}\setminus\mathcal{C}_r$.
        Any pair of points in $S$ in different subpolygons of $\supp{\Lambda^{-1}r}\setminus\mathcal{C}_r$ is at least $\Lambda^{-2}r$ apart by property~\ref{itm:not_intersect_nerve} of the cuts in $\mathcal{C}_r$ and~\cref{obs:pairwise-component-distance}:
        since these cuts do not intersect $\supp{\Lambda^{-2}r}$, any pair of points from distinct subpolygons of $\supp{\Lambda^{-1}r}\setminus\mathcal{C}_r$ must be in different polygons in $\supp{\Lambda^{-2}r}$.
        As the geodesic diameter of a junction in $\J_{\Lambda^{-1} r}$ is bounded by $18r$, and points in a non-trivial junction-star (that is, involving more than one point) have distance at most~$3r$ to the junction, there are at most $16\frac{(24 r)^2}{(\Lambda^{-2}r)^2}=16\cdot24^2\Lambda^4$ stars per junction with at least one point in~$D$.
        It remains to bound the number of junctions in $\J_{\Lambda^{-1}r}$ whose stars contribute to $\pi_T(p,q)$. Any such junction $J$ has a star point in $D$, and as this point lies within geodesic distance $3r$ of~$J$ and~$J$ has geodesic diameter at most $18r$, every point of $J$ lies in $D_P(p,\delta\Lambda^2r+21r)\subseteq D(p,2\delta\Lambda^2 r)$. Choosing any point of $S$ in each such junction, these points are pairwise at least $\frac{\Lambda^{-1}r}{2}$ apart by the same argument as for scale $r$, so by \Cref{obs:pack} there are at most $16\frac{(2\delta\Lambda^2r)^2}{(\Lambda^{-1}r/2)^2}=256\delta^2\Lambda^6$ such junctions in $\J_{\Lambda^{-1}r}$. The maximum edge length of an edge in such a junction-star is $24 r$ by the {\hyperref[inv:pack]{Packedness}} invariant. We continue similarly for the smaller scales, and obtain:
        \begin{compactenum}
            \item Scale $r$: $ 64\delta^2\Lambda^4$ junctions, each contributing a single star with maximum edge length~$24\Lambda r$.
            \item Scale $\Lambda^{-1}r$: $ 256\delta^2\Lambda^6$ junctions, each contributing $16\cdot24^2\Lambda^4$ stars with maximum edge length~$24 r$.
            \item Scale $\Lambda^{-2}r$: $ 256\delta^2\Lambda^8$ junctions, each contributing $16\cdot24^2\Lambda^4$ stars with maximum edge length $24\Lambda^{-1}r$.
            \item Scale $\Lambda^{-3}r$: $ 256\delta^2\Lambda^{10}$ junctions, each contributing $16\cdot24^2\Lambda^4$ stars with maximum edge length $24\Lambda^{-2}r$.
            \item All scales $\leq\Lambda^{-4}r$: At most $n$ star-edges of length at most $24\Lambda^{-3}r$.
        \end{compactenum}

        Next, we consider the face-stars added in the construction phase for scale $r$.
        By the {\hyperref[inv:pack]{Packedness}} invariant, we can select for each edge of $\pi_T(p,q)$ that is associated with such a face-star an endpoint such that all selected points lie in distinct weakly simple polygons in $\supp{\Lambda^{-1}r}$. This allows us to directly bound the number of edges of $\pi_T(p,q)$ from such face-stars intersecting $D$, instead of the number of face-stars intersecting $D$.
        Hence,~\cref{obs:pack,obs:pairwise-component-distance} bound the number of face-star edges participating in $\pi_T(p,q)$, by~$16\frac{(\delta\Lambda^2r)^2}{(\Lambda^{-1}r)^2}=16\delta^2\Lambda^6$.
        We similarly obtain for the smaller scales:
        \begin{compactenum}
            \item Scale $r$: $16\delta^2\Lambda^6$ face-star edges with maximum edge length $61(\Lambda^2\log n)r$.
            \item Scale $\Lambda^{-1}r$: $16\delta^2\Lambda^8$ face-star edges with maximum edge length $61(\Lambda\log n)r$.
            \item Scale $\Lambda^{-2}r$: $16\delta^2\Lambda^{10}$ face-star edges with maximum edge length $61(\log n)r$.
            \item Scale $\Lambda^{-3}r$: $16\delta^2\Lambda^{12}$ face-star edges with maximum edge length $61(\Lambda^{-1}\log n)r$.
            \item Scale $\Lambda^{-4}r$: $16\delta^2\Lambda^{14}$ face-star edges with maximum edge length $61(\Lambda^{-2}\log n)r$.
            \item All scales $\leq\Lambda^{-5}r$: at most $n$ other face-star edges of length at most $61(\Lambda^{-3}\log n)r$.
        \end{compactenum}

        Since from every star, at most two edges can participate in a shortest path in $T$, we obtain a crude bound on the total length of a shortest path between
        $p$ and $q$ of:
        \begin{align*}
            \dist_T(p,q)
            &\leq
            \underbrace{2\cdot64\delta^2\Lambda^4 \cdot 24\Lambda r}_{\substack{\text{junction-stars}\\\text{from scale $r$}}}
            +\underbrace{\left(\sum_{s=1}^{3}
                             2\cdot\overbrace{256\delta^2\Lambda^{4+2s}}^{\text{$\#$ of junctions}}
                             \cdot\overbrace{16\cdot24^2\Lambda^4}^{\text{stars per junction}}
                             \cdot\overbrace{24\Lambda^{1-s}r}^{\text{length}}\right)}_{\substack{\text{junction-stars from scales}\\\text{$\Lambda^{-s}r$ $=$ $\Lambda^{-1}r$, $\Lambda^{-2}r$, $\Lambda^{-3}r$}}}
            +\underbrace{n\cdot24\Lambda^{-3}r}_{\substack{\text{junction-stars}\\\text{from scales $\leq\Lambda^{-4}r$}}}\\
            &\qquad+
            \underbrace{\left(\sum_{s=0}^{4}
                            \overbrace{16\delta^2\Lambda^{6+2s}}^{\text{$\#$ of face-star edges}}
                            \cdot\overbrace{61 \log  n \cdot \Lambda^{2-s}r}^{\text{length}}
            \right)}_{\substack{\text{face-stars from scales}\\\text{$\Lambda^{-s}r$ $=$ $r$, $\Lambda^{-1}r$, $\Lambda^{-2}r$, $\Lambda^{-3}r$, $\Lambda^{-4}r$}}}
            +\underbrace{n\cdot61\log  n \cdot \Lambda^{-3}r}_{\substack{\text{face-stars}\\\text{from scales $\leq\Lambda^{-5}r$}}}\\
            &\leq
            \delta^2\left(
                        2\cdot64\cdot24
                +3\cdot2\cdot256\cdot16\cdot24^3
                +5\cdot16\cdot61\log n
            \right)\Lambda^{12}r+
            \left(24+61\log n\right)n\Lambda^{-3}r\\
            &=
            \delta^2\left(
                        339\,741\,696+4\,880\log n
            \right)\Lambda^{12}r
            +
            \left(24+61\log n\right)n\Lambda^{-3}r.
        \end{align*}
        Here, we assumed that each of the scales $\Lambda^{-4}r, \Lambda^{-3}r, \Lambda^{-2}r, \Lambda^{-1}r$ is indeed present in the compressed scale sequence. If this is not the case, then by the properties of the compressed scale sequence the decrease in size between such scales is at least a factor $\Lambda^{10}$, and we are thus straight away in the~$\leq \Lambda^{-5}r$ case. So, the stated bound is indeed an upper bound on the length of the path.
        Hence, the dilation of any pair $p,q$ whose Euclidean distance is in $[r/(2\Lambda^2),2r]$ is at most

        \begin{align*}
            \dil_T(p,q)
            &=\frac{\dist_T(p,q)}{\|p-q\|}
            \leq \frac{\dist_T(p,q)}{r/(2\Lambda^2)}\leq
            2\delta^2\left(339\,741\,696+4\,880\log n\right)\Lambda^{14}
            +
            2\left(24+61\log n\right)\frac{n}{\Lambda}\\
            &\leq
            2\left[
                 \delta^2\left(339\,741\,696+4\,880\log n\right)
                 +24+61\log n
            \right]
            \max\left\{
                    \frac{n}{\Lambda},\Lambda^{14}
            \right\}\\
            &\leq
            2\left[
                 (200\log n+10)^2
                 \left(339\,741\,696+4\,880\log n\right)
                 +24+61\log n
            \right]
            \max\left\{
                    \frac{n}{\Lambda},\Lambda^{14}
            \right\}\\
            &\leq
            29\,965\,648\,003\,370\,\log^3 n
            \max\left\{
                    \frac{n}{\Lambda},\Lambda^{14}
            \right\}.\qedhere
        \end{align*}
    \end{proof}

    \subsection{Putting everything together}\label{sec:putting_everything_together}

    Finally, we combine the ingredients from the previous sections to show that our algorithm is indeed a polynomial time $o(n)$-approximation algorithm for the minimum dilation tree problem.

    \main*
    \begin{proof}
        Choose $\Lambda=\max\{6,n^{1/15}\}$.
        If the algorithm terminates before or during the cutting phase, then $\opt\geq \Lambda/100$. Since the Euclidean minimum spanning tree has dilation at most $n-1$, the returned tree is a $100n/\Lambda=O(n^{14/15})$ approximation. Otherwise, the {\hyperref[inv:conn]{Connectedness}} invariant implies that the returned graph is indeed a spanning tree, as $\mathcal{C}_{\Lambda r}=\emptyset$ at the largest scale $r=r_{\max}$. Furthermore, at this scale, $\supp{r}$ is the convex hull of $S$, and hence $\dist_{\supp{r}}(p,q)=\|p-q\|<r$ for every pair $p,q\in S$. Thus, the {\hyperref[inv:dil]{Dilation}} invariant applies to all pairs, implying that
        \[\dil(T) \leq 29\,965\,648\,003\,370\,\log^3 n \max\left\{\frac n\Lambda,\Lambda^{14}\right\}=O(n^{14/15}\log^3 n).\]
        As $\opt\geq1$, this is an $O(n^{14/15}\log^3 n)$ approximation.

        Finally, we bound the running time. A consistent sequence is constructed in~$O(n^5\log n)$ via \Cref{thm:construct_consistent_sequence}. For the cutting phases, observe that at any fixed scale $r$, we can charge each nontrivial
        cut computation to a distinct point of $S$ in the
        corresponding face. Hence, at most $n$ cuts are computed per scale.
        By \Cref{lem:construct_min_realization}, each cut takes~$O(n^2\log n)$ time. Since there are~$O(n^2)$ scales,
        the cutting phase takes~$O(n^5\log n)$ time in total. Lastly, each individual construction phase can be easily bounded by~$O(n^3)$. Hence the total running time is in~$O(n^5\log n)$.
    \end{proof}

    \begin{corollary}
        If the point set $S$ has spread $\Delta$, there is an $O(\min\{\Delta,n^{14/15}\log^3 n\})$-approximation algorithm with polynomial running time. In particular, if $\Delta=n^\delta$, then the approximation ratio is $O(\min\{n^\delta,n^{14/15}\log^3 n\})$.
    \end{corollary}
    \begin{proof}
        Construct both the tree $T$ of~\Cref{thm:main} and a star $T'$ centred at an arbitrary point of~$S$, and return the one with smaller dilation. Let $\dist_{\min}$ and $\dist_{\max}$ be the minimum and maximum Euclidean distance between any pair of points in $S$. For any two points $p,q$, their distance in $T'$ is at~most~$2\dist_{\max}$, while $\|p-q\|\geq \dist_{\min}$. Hence $T'$ has dilation at most $2\dist_{\max}/\dist_{\min}=2\Delta= O(\Delta\cdot \opt)$. \Cref{thm:main} implies $T$ has a dilation in $O(n^{14/15}\log^3 n\cdot\opt)$. Returning the better of the two trees thus gives an $O(\min\{\Delta,n^{14/15}\log^3 n\})$-approximation.
    \end{proof}

    \section{Conclusion}\label{sec:concluding_remarks}

    In this paper, we have broken the $\Omega(n)$-barrier on the approximability of the minimum dilation tree for arbitrary point sets in $\mathbb{R}^2$, resolving a long-standing open question about the minimum dilation tree. Along the way, we categorized
    sets of points whose minimum dilation tree has dilation at least $\Omega(n^{14/15}\log ^3 n)$. In particular, there is a constant $C$, such that any set of points whose minimum dilation tree has dilation at least $C n^{14/15} \log^3 n$ must have either a circle or a comb obstruction, as otherwise our algorithm outputs a tree of dilation $O(n^{14/15}\log^3 n)$.

    For ease of exposition, we have in several places bounded our expressions pessimistically. We expect the true running time of our algorithm to be slightly better than $\tilde{O}(n^5)$. One can, for example, reduce the size of the compressed scale sequence to $O(n)$, via well-separated pair-decompositions in~$\mathbb{R}^2$. Similarly, the computation of the spine-cutting graph across all faces can be done in sub-cubic time. Pushing the overall running time of our algorithm past $\tilde{O}(n^3)$ seems plausible, but optimizing the running time was not the focus of the paper.

    Similarly, the presented approximation ratio for our algorithm is unlikely to be tight. In fact, using even slightly better bounds in the proof of \Cref{lem:dilation}, together with a choice of ${\Lambda = \left(\frac{n}{\log^3n}\right)^{1/15}}$, one can already obtain an approximation ratio of $O(n^{14/15}\log^{1/5}n)$. Moreover, a better bound on how close the junction- and face-stars can be would lead to a further improved approximation ratio. However, a natural barrier for our overall approach is an approximation ratio of $O(\max\{n/\Lambda,\Lambda\})$, which translates to a $\Theta(\sqrt{n})$-approximation, similar to the comb construction in \Cref{sec:tech_overview}. Breaking this barrier, or even obtaining a bound close to this, will most likely require novel insights and techniques.

    The large constant of $29\,965\,648\,003\,370$ is also pessimistic, although pushing it to be below $1000$ seems infeasible with our current techniques. An unfortunate consequence of this large constant is that our algorithm is far from practical. The smallest $n$ for which our algorithm has a better approximation guarantee
    than the minimum spanning tree is $n\gtrsim 10^{340}$, since $n$ needs to satisfy $29\,965\,648\,003\,370 \,n^{14/15}\log^3n\leq n-1$.

    A natural open question that arises from our results is: Can these techniques be generalized to
    higher dimensions? Even the $\alpha$-tree partition, our most basic algorithmic tool, seems to require a stronger structural result than a three-dimensional version of the circle obstruction.
    Another natural extension, past higher-dimensional Euclidean spaces, is the question of whether these techniques can also be extended to doubling spaces, forgoing all Euclidean geometry.

    \subsubsection*{Acknowledgements}
    The second, fourth and fifth authors would like to thank Kevin Buchin, Benedikt Kolbe, Abhiruk Lahiri, Martin P. Seybold, Marco Ricci, the Rhine Ruhr Computational Geometry Workshop, Carolin Rehs, Torben Scheele, and Jan Titzeck for interesting discussions on this problem.
    \smallskip

    Funding in direct support of this work:
    \begin{compactitem}
        \item Sarita de Berg is supported by Danmarks Frie Forskningsfond via the grant ``Dynamic Graphs: Distributed and Geometry'' with grant-ID 10.46540/4251-00004B.
        \item Jacobus Conradi is funded by the Carlsberg Foundation, grant CF24-1929.
        \item Peter Kramer is supported by the Deutsche Forschungsgemeinschaft (DFG, German Research Foundation), grant 530918134.
        \item André Nusser is supported by the France 2030 investment plan managed by the ANR as part of the Initiative of Excellence of Université Côte d'Azur with reference number ANR-15-IDEX-01.
        \item Sampson Wong is supported by the European Union's Marie Skłodowska-Curie Actions Postdoctoral Fellowship, grant 101146276.
    \end{compactitem}

    \subsubsection*{AI Methodology}
    All results in this paper were obtained without the use of AI. Only in the final stages did we use AI tools to polish the writing and to flag potential logical inconsistencies in the exposition; every such suggestion was checked and, where adopted, implemented by the authors. 

    \bibliography{references}
    \appendix\newpage

    \section{Omitted proofs}
    \label[appendix]{apx:omitted proofs}

    \lemThinObstructionCheck*
    \begin{proof}
        We simply enumerate all distinct $r$-supports of $S$ and check if the polygons are $\Lambda{r}$-thin.
        By~\cref{obs:number-of-distinct-supports}, the number of distinct $r$-supports is bounded by the number of edges in a Delaunay triangulation and hence, by $3n$.

        Observe that the radius of the largest disk contained in a \emph{weakly} simple polygon (its inradius) is equal to the maximum across all inclusion-maximal simple subpolygons of it.
        The inradius of a simple polygon $P$ can be computed in linear time using a result of Chin, Snoeyink, and Wang~\cite{DBLP:journals/dcg/ChinSW99}, as a disk of radius matching exactly the inradius of $P$ can be placed at one of the linearly many vertices of the medial axis of $P$. Hence, it is sufficient to identify the inclusion-maximal simple subpolygons of the distinct $r$-supports of $S$ and verify that these are $\Lambda{r}$-thin.

        The Delaunay triangulation of $S$ can be computed in ${O}(n\log n)$ time using classical methods~\cite{DBLP:journals/algorithmica/Fortune87}.
        Sorting its edges by increasing length in $O(n\log n)$ time allows us to efficiently enumerate all distinct subsets of edges with bounded length, that is, the sets $\mathrm{Delaunay}_{\leq r}$.
        Recall that their number is bounded by $3n$.
        It remains to compute the maximal simple subpolygons of a fixed $r$-support.
        We can do this in linear time by walking along the outer face from an extremal point in, e.g., counterclockwise orientation until we find a cycle.

        The total runtime is then $O(n\log n)$ for preprocessing, $O(n)$ for the enumeration, and~$O(n)$ for the inradius computation of each simple polygon, i.e., $O(n^2)$. Finally, if $S$ is not $\Lambda$-thin, then there is some $r$ for which a weakly simple polygon in the $r$-support is not $\Lambda r$-thin, and thus contains a disk of radius $\Lambda r$. \Cref{cor:circle-test} then directly implies that $\opt \geq 2\Lambda$.
    \end{proof}

    \begin{restatable}{lemma}{floodfill}\label[lemma]{lem:floodfill}
        Let $P$ be a $\Lambda r$-thin weakly simple polygon, $p$ a point in $P$, and $R\ge \Lambda r$. Let $\gamma$ be a connected component of $\partial D_P(p,R)\cap \int(P)$. If~$a,b$ are two points on $\gamma$, then
        $\dist_P(a,b)\leq 6\Lambda r$.
    \end{restatable}
    \begin{proof}
        See \Cref{fig:wavefront_appendix} for an illustration of the proof. The illustration shows $a$ and $b$ as the endpoints of $\gamma$, but all of the following arguments go through for any choice of $a,b$ on $\gamma$. For a point $x$ on~$\gamma$ between $a$ and $b$, consider the (Euclidean) disk $D_x$ centred at point $p_x$ at distance~$\Lambda r$ along~$\pi_P(x,p)$ with radius~$\Lambda r$. We first consider the case that there is an $x \in \gamma$ between $a$ and $b$ for which~$p_x \in \pi_P(p,a) \cap \pi_P(p,b)$. As the points $x$, $a$, and $b$ all have geodesic distance $R$ from $p$, and all three shortest paths from $p$ to these points overlap up to $p_x$, it must be that $ \dist_P(p_x,a) = \dist_P(p_x,b) = \dist_P(p_x,x) = \Lambda r$. By triangle inequality we then have $\dist_P(a,b) \leq \dist_P(a,p_x) + \dist_P(p_x,b) = 2\Lambda r$.

        Now, consider the case where $p_x \notin \pi_P(p,a)\cap \pi_P(p,b)$ for any $x \in \gamma$ between $a$ and $b$. Consider the two shortest paths $\pi_P(p,a)$ and $\pi_P(p,b)$ that form a funnel in $P$. Let $c$ be the last point in~${\pi_P(p,a) \cap \pi_P(p,b)}$, called the apex of the funnel. We denote by~$A$ the closed subset of~$P$ bounded by $\pi_P(a,c)$, $\pi_P(c,b)$, and the part of $\gamma$ between $a$ and $b$.
        Let~$\overline{D_x}$ be the subset of $D_x$ contained in~$A$, i.e., $\overline D_x = D_x \cap A$, see \Cref{fig:wavefront_appendix}.

        \begin{figure}
            \centering
            \includegraphics{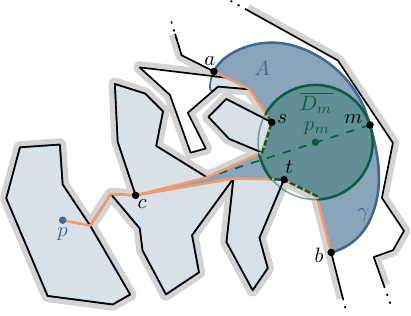}
            \caption{The geodesic disk $D_P(p,R)$ in blue with the region $A$ in dark blue. The area $\overline{D_m}$ is the green shaded subset of the disk $D_m$ intersected with $A$.}
            \label{fig:wavefront_appendix}
        \end{figure}

        We first show that for any
        point $q \in (\overline{D_x} \cap \partial P)$, either $q$ is on $\pi_P(c,a) \cup \pi_P(c,b)$, or the last vertex of the shortest path $\pi_P(q,p_x)$ is on $\pi_P(c,a) \cup \pi_P(c,b)$.
        Suppose for contradiction that~$q$ is the point geodesically closest to $p_x$ for which this does not hold. Then, because $q$ is the closest point to~$p_x$ for which this does not hold, there is no other point of $\partial P$ on $\pi_P(q,p_x)$.
        So,~${\dist_P(q,p_x) = |qp_x| \leq \Lambda r}$. By choice of $p_x$, $\dist_P(p_x,p) = R - \Lambda r$ and thus $\dist_P(q,p) \leq \dist_P(p_x,p) + |qp_x| \leq R$, i.e., $q$ is in~$D_P(p,R)$.
        As $q \in \overline D_x$, it must be that $q$ lies in $A$. So, there must be a point of $\partial P$ not on either of the shortest paths $\pi_P(c,a)$ or $\pi_P(c,b)$ that intersects $\gamma$. This contradicts $\gamma$ being a connected component of~$\partial D_P(p,R)\cap \int(P)$.

        Next, we prove that there is a point $m \in \gamma$ between $a$ and $b$ such that $p_m$ sees a point on~${\overline D_m \cap \pi_P(c,a)}$ and a point on~$\overline D_m \cap \pi_P(c,b)$. Because~$P$ is $\Lambda r$-thin, there must be for any disk $D_x$ a point $q$ of $\partial P$ contained in~$D_x$. Consider the segment $p_x q$. This segment either (1) is contained in~$\overline D_x$, (2) intersects~$\pi_P(c,a) \cup \pi_P(c,b)$, or (3) intersects $\gamma$.
        In case (1),  $q$ is in $A$, and thus the existence of $q$ implies the last vertex of the shortest path $\pi_P(q,p_x)$ is on~$\pi_P(c,a) \cup \pi_P(c,b)$. This vertex is also contained in $\overline{D_x}$ and visible to $p_x$.
        In case (2), the closest point of $p_x q$ in~${(\pi_P(c,a) \cup \pi_P(c,b)) \cap \overline D_x}$ is visible to $p_x$.
        In case (3), consider the point $y$ that is on $p_x q \cap \int(\gamma)$. Any vertex of $\partial P$ on $\pi_P(p_x,y)$ would again imply a point on $\pi_P(c,a) \cup \pi_P(c,b)$ visible to $p_x$ as before. If there is no such vertex, then~$\dist_P(p_x,y) = |p_x y|$. Because $y \in \gamma$, we have that~${\dist_P(p,y) \leq \dist_P(p,p_x) + \dist_P(p_x, y) = R - \Lambda r + |p_x y|}$. So, $|p_x y| \geq \Lambda r$. However, by definition of~$q$, we have~$|p_x q| \leq \Lambda r$ and $q \neq y$, a contradiction.

        We conclude that there is a point of $\pi_P(c,a) \cup \pi_P(c,b)$ in $\overline D_x$ visible to $p_x$.
        When sliding the point~$x$ along $\gamma$ from $a$ to $b$, the point $p_x$ moves continuously as well because $P$ is a weakly simple polygon.
        During this sliding, there must be a time when $p_x$ sees a point in both $\pi_P(c,a)\cap \overline D_x$ and~$\pi_P(c,b) \cap \overline D_x$, as $p_x$ must see a point on $(\pi_P(c,a)\cup \pi_P(c,b)) \cap \overline D_x$ at any time.

        Let $\overline D_m$ be such a disk where $p_m$ sees a point $s$ on $\pi_P(c,a) \cap \overline D_m$ and a point $t$ on $\pi_P(c,b) \cap \overline D_m$. We have that $\dist_P(s,p_m) \leq \Lambda r$ and $\dist_P(t,p_m) \leq \Lambda r$. So, by triangle inequality $\dist_P(s,t) \leq 2 \Lambda r$. Furthermore,  since $\dist_P(p,p_m) = R - \Lambda r$, we obtain by triangle inequality that $\dist_P(p,s) \geq R - 2 \Lambda r$ and so $\dist_P(s,a) \leq 2\Lambda r$. Symmetrically, $\dist_P(t,b) \leq 2\Lambda r$. We conclude that $\dist_P(a,b) \leq \dist_P(a,s) + \dist_P(s,t) + \dist_P(t,b) \leq 6 \Lambda r$.
    \end{proof}

    \treePartition*

    \begin{proof}
        As $S$ is $\Lambda$-thin, $P$ is $\Lambda r$-thin.
        Furthermore, the boundary edges of $P$ have length at most $r$ by definition.
        We fix an arbitrary vertex $p$ of $P$ and define a family of geodesic disks~${D_i:=D_P(p,R_i)}$ for $i\in\mathbb{N}_+$, each centred at~$p$ with radius $R_i$, as shown in \Cref{fig:tree-partition-ap}(b).
        We would like to avoid,
        for any vertex $q$ of $P$, a coincidence of its geodesic distance to $p$ with the radius of a disk $D_i$.
        That is, we want that $R_i\neq \dist_P(p,q)$ for all choices of $i$ and $q$.
        There are finitely many points to avoid, so we can choose a value $\mu \in [0,\Lambda r)$ such that $R_i\coloneqq 4 i\Lambda r -\mu$ does not coincide with the geodesic distance from $p$ to any other vertex, for any choice of $i$.
        Finally, note that all disks coincide for $i\geq n$, as the diameter of $P$ is at most~$nr$.
        \begin{figure}%
            \centering%
            \includegraphics{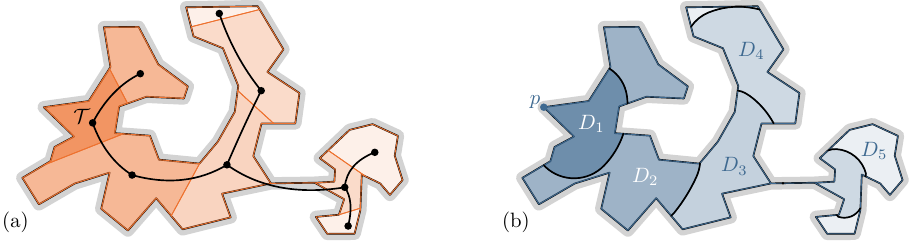}%
            \caption{(a) An $\alpha$-tree partition and the corresponding intersection graph. (b) The family of geodesic disks used in the construction of an $\alpha$-tree partition due to~\cref{thm:tree-partition}.}
            \label{fig:tree-partition-ap}
        \end{figure}

        The boundary $\partial D_i$ of the disk $D_i$
        contains finitely many inclusion-maximal curves consisting of points with geodesic distance exactly $R_i$ to $p$.
        Note that any such curve $\gamma\subset \partial D_i$ may consist of a single point, as $P$ is weakly simple.
        For brevity, we will call~$\gamma$ a \emph{level curve} of $\partial D_i$ for the remainder of this proof.
        Let~${u,v\in\partial P}$ denote the endpoints of $\gamma$, with $u=v$ exactly if~$\gamma$ is a singleton.
        We define the \emph{portal} of $\gamma$ as the geodesic path $\overline \gamma = \pi_P(u,v)$.

        By~\Cref{lem:floodfill}, $\overline \gamma$ has length at most~$6\Lambda{r}$.
        This means that for any point on the portal, the geodesic distance to one of the endpoints of the portal is at most $3\Lambda r$, so the geodesic distance from any point on $\overline \gamma$ to $p$ is at least $R_i - 3 \Lambda r$.
        We conclude that for any $j<i$, the portal of~$\gamma$ does not intersect~${D_j\subset D_P(p,R_i - 3\Lambda r)}$.
        Because $D_i$ is geodesically convex with respect to~$P$~\cite[Corollary~1]{DBLP:journals/dcg/PollackSR89}, it further follows that $\overline \gamma$ is contained in~$D_i$.
        Hence, portals from different disks $D_i$ and $D_j$ are disjoint.
        Next, consider portals from the same disk~$D_i$.
        If two level curves~$\gamma$ and~$\gamma'$ have endpoints $u,v$ and $u',v'$ respectively, observe that the funnels~$A=\gamma \oplus \pi_P(p,u) \oplus \pi_P(p,v)$ and~$A'=\gamma' \oplus \pi_P(p,u') \oplus \pi_P(p,v')$ are interior-disjoint.
        It follows that the portals $\overline \gamma$ and $\overline {\gamma'}$ cannot cross, leaving the possibility of a shared line segment or point in $\overline \gamma \cap \overline {\gamma'}$.
        Such a line segment would have to lie on the shared boundary of the funnels $A$ and $A'$, i.e., on a shortest path between $p$ and one of the four endpoints.
        This line segment cannot end at any portal endpoint, as interior points of shortest paths between points on level curves cannot lie on level curves.
        Therefore, an inclusion-maximal shared line segment would have to end at a vertex of~$P$ that is shared by the portals.
        Similarly, if the portals share only an isolated point, this point must be a vertex of~$P$, as the portal paths are geodesic.
        Hence, assume that a vertex $w$ of $P$ lies in the interior of both portals.
        Note that this already rules out intersections involving singleton portals.
        Then the interior angle of both funnels $A$ and $A'$ at $w$ has to be at least $\pi$.
        Since the funnels are interior-disjoint, this requires an interior angle of at least~$2\pi$ in $P$, a contradiction.
        We conclude that $\overline{\gamma}$ and $\overline{\gamma'}$ are disjoint.

        For the remainder of this proof, denote by~$\Gamma_i$ the set of portals induced by the maximal level curves of~$\partial D_i$.
        We say that the portals in~$\Gamma_{i+1}$ are children of portals in~$\Gamma_{i}$. Note that the geodesic distance between a parent and child portal is~$\geq \Lambda r$.

        The union of $\Gamma_i$ across all $1\leq i\leq n$ will define our partition, as illustrated in~\Cref{fig:tree-partition-ap}(a). The portals in $\bigcup_{i=1}^n \, \Gamma_i$ divide the polygon $P$ into subpolygons, which we call pieces. Each partition piece is geodesically convex with respect to $P$: any shortest path leaving the piece via a portal has to re-enter via the same portal, and can thereby be shortcut by following the portal.
        Next, we will prove that these pieces satisfy the three properties stated in \cref{def:tree-partition}. After that, we will prove that the pieces satisfy the distance preserving property in \cref{definition:distance_preserving}.

        First, by definition, the pieces form a partition of~$P$. Second, to show that the intersection graph is a tree, consider a portal separating two pieces. This portal must also separate the polygon~$P$ into two connected components, otherwise, the polygon~$P$ would not have been weakly simple. It follows that the intersection graph is a tree.
        Third, to bound the diameter, consider some piece bounded by one portal in $\Gamma_i$, which is the parent portal, generated by a curve $\gamma$, and multiple portals in~$\Gamma_{i+1}$. By the geodesic convexity and the fact that any portal of $D_{i+1}$ does not intersect $D_i$, all of $\gamma$ is contained in the piece. To bound the diameter of the piece, consider two points $a,b$ contained in the piece. We define a point $a^*$ on $\gamma$ as follows.
        If~$a$ has geodesic distance smaller than $R_i$ to $p$, then the intersection point of $\pi_P(a,p)$ and~$\Gamma_i$, which is contained in $D_P(p,R_i) \setminus D_P(p,R_i - 3 \Lambda r)$, is within geodesic distance $3 \Lambda r$ of $a$. By \Cref{lem:floodfill} the distance from this point to one of the endpoints of $\gamma$ is also bounded by~$3\Lambda r$.
        We set $a^*$ to be this endpoint. If $a$ has geodesic distance at least $R_i$ from $p$, then the path $\pi_P(a,p)$ must intersect $\gamma$.
        In particular, as $a$ is contained in $D_{i+1}$, the point has geodesic distance at most $4\Lambda r$ of $a$. Set $a^*$ to be this point. Similarly define $b^*$.
        The diameter of this piece is at most~$\dist_P(a,a^*) + \dist_P(a^*,b^*) + \dist_P(b^*,b) \leq (6+6+6)\Lambda r = 18\Lambda r$, where we again apply \Cref{lem:floodfill}. The only piece for which this argument does not hold is the one containing $p$, for which the diameter is trivially~$8\Lambda r\leq 18\Lambda r$.
        Putting this together, we conclude that the pieces form an~$(18 \Lambda r)$-tree partition of~$P$.

        Next, we prove the distance-preserving property (\Cref{definition:distance_preserving}). Let $s\in N_i$ and $t\in N_j$ be two points in $P$ contained in pieces $N_i$ and $N_j$ respectively.
        The shortest path $\pi_P(s,t)$ traverses~$\dist_\T(N_i,N_j)-1$ distinct pieces of the partition, each from one portal to another. Of these pieces, all but the lowest common ancestor piece of $N_i$ and $N_j$ are crossed by a portal-portal path that connects a parent and a child. So the geodesic distance~$\dist_P(s,t)$ is at least~${(\dist_\T(N_i,N_j)-2)\Lambda r}$. Note that if $N_i$ is an ancestor of $N_j$, or the other way around, the stronger bound of~$(\dist_\T(N_i,N_j)-1)\Lambda r$ applies.
        Hence, our tree partition is $18$-distance-preserving.

        It remains to argue the running time of our algorithm. We first partition the weakly simple polygon $P$ into simple polygons connected by polygonal paths. We preprocess these simple polygons and the connecting paths for efficient point location queries in $O(n\log n)$ time~\cite{Kirkpatrick_point_location}. Furthermore, we can preprocess the connecting paths in linear time such that we can find the length of a shortest path from one simple polygon to another in constant time.
        Next, we compute a shortest path data structure for each simple polygon in~$P$ in $O(n\log n)$ time~\cite{2PSP_simple_polygon}.
        With it, we can determine for any two points~${u,v\in P}$ the geodesic distance~$\dist_P(u,v)$ in~$O(\log n)$ time, and the shortest path~$\pi_P(u,v)$ in additional output-sensitive linear time, that is, linear in the number of segments defining it.
        Note that for points in different simple subpolygons of~$P$, we first compute the shortest path distance between those polygons, and add to that the distance to and from the exit and entry points of those subpolygons to obtain the full geodesic distance.
        A similar process allows us to obtain the corresponding path in output-sensitive time.
        Additionally, we construct the shortest path map for~$p$, which partitions $P$ into cells for which the shortest path to $p$ visits the same vertices, in $O(n \log n)$ time~\cite{DBLP:journals/siamcomp/HershbergerS99}. This map has linear complexity~\cite{DBLP:journals/siamcomp/HershbergerS99}.

        We now can compute all endpoints of the
        level curves of $\partial D_i$
        in~$O(n\log n)$ time as follows.
        Recall that the maximum edge length of $P$ is $r$. Thus, no edge of $P$ can contain arc endpoints from different disks $D_i$ as a direct result of the growth in radius of the disks. By convexity of the geodesic distance function along an edge, there can be at most two endpoints from $D_i$ on a single edge.
        We can thus find all endpoints by walking around the boundary of $P$, checking the distance to $p$ whenever we encounter the boundary of a shortest path map cell, and computing the possible endpoint on this segment if necessary.
        Going around the boundary of the weakly simple polygon~$P$ with $n$ vertices once visits at most $2n-2$ edges.
        It follows that there are at most $4n-4$ portal endpoints, and hence $2n-2$ portals in total.
        This, in conjunction with the fact that every portal lies in the intersection of exactly two pieces, implies that there are at most $2n-1$ pieces in the resulting partition.

        We then compute the shortest paths connecting corresponding endpoints using the shortest path data structure. While the complexity of a shortest path may be $n$, the paths we compute are pairwise disjoint. Hence, we can compute \emph{all} these shortest paths in~$O(n \log n)$ total time.

        Finally, assembling each weakly simple subpolygon can be done in linear time, as every portal exists on the boundary of exactly two pieces, and the respective incidence relation can be derived by considering their order along the boundary of $P$.
        Note that it is possible for portals to partially coincide with $\partial P$, which is why this process (possibly) yields weakly simple polygons.
    \end{proof}

    \lemmaMaxEdgeLengthSupport*

    \begin{proof}
        \phantomsection\label{prf:lem:max-edge-length-support}
        For a contradiction, let $P\in\supp{r}$ be $\Lambda r$-thin and contain an edge $e$ of the Delaunay triangulation with~${\len(e)\geq 3\Lambda r}$.
        As $e$ is a Delaunay edge, there is a disk $D(p,R)$ with~${R\geq \len(e)/2}$ whose boundary contains both endpoints of $e$ and whose interior contains no point of~$S$.

        If $D(p,R)\subset P$, we immediately arrive at a contradiction as $R\geq {1.5\Lambda r}$, while $P$ is $\Lambda r$-thin.
        Assume therefore that the disk is not fully contained in $P$.
        By $\len(e)>r$, we have that $e$ is not a boundary edge, and its midpoint $m$ lies in the interior of~$P$.
        Hence, $D(p,R)$ intersects~$P$.

        We show that there is a similarly-sized concentric disk that is fully contained in~$P$.
        To this end, let $f$ be any line segment of~$\partial P$ that intersects~$D(p,R)$.
        Then~$f$ starts and ends outside of~$D(p,R)$ and has $\len(f)\leq r$.
        Let~$x$ be the midpoint of the intersection of $f$ and~$D(p,R)$, which is then also the closest point on $f$ to~$p$.
        By applying the Pythagorean theorem to a right triangle spanned by~$x$,~$p$, and either endpoint of the intersecting portion of $f$, we obtain that the distance between~$p$ and~$f$ is at least $R'=\sqrt{R^2-(r/2)^2}\geq \sqrt{{(1.5\Lambda r)^2}-(r/2)^2}>\Lambda r$.
        The concentric disk $D(p,R')$ hence can touch, but not intersect $\partial P$, and it remains to show that it lies inside~$P$.
        To this end, consider again the line segment $e$ and its midpoint $m$.
        By the Pythagorean theorem, the distance between $m$ and $p$ is exactly~$\sqrt{R^2-(\len(e)/2)^2}$, which is strictly less than $R'$ by $\len(e)>r$.
        Hence, $m$ lies in the smaller concentric disk, which must therefore be contained fully inside $P$.
        We again obtain a contradiction to the thinness of $P$.
    \end{proof}

    \subsection{Structuring the point set via the spine-cutting graph}

    This section contains the omitted proof of~\cref{lem:construct_min_realization}.
    For the reader's convenience, we recall the following two definitions here.
    \defSpineCuttingGraph*
    \defGoodCut*

    \lemmaFindCutOrLowerBound*

    \begin{proof}
        Recall that each relevant vertex of~$G_{F,\spineset}$ corresponds to a region; either to the outer face or to a connected component of~$F\setminus\spineset$ that intersects at least two junctions.
        A good cut of~$G_{F,\spineset}$ must connect all these regions (property~\ref{itm:tree}).

        We construct an auxiliary graph $H$, \cref{fig:spine_cutting_ap} illustrates an example.
        The vertices of~$H$ are those of $G_{F,\spineset}$ together with a vertex for each triangle of the Delaunay triangulation of~$S$ that meets three criteria:
        it (i) is contained in an $F$-spine of $\spineset$,
        (ii)~does not intersect the interior of any $F$-nerve, and
        (iii)~does not intersect the $z_2$-buffer of~$F$ for~${z_2 = r/2}$.
        As each $F$-spine~$\spine\in\spineset$ is weakly simple, part of its boundary may form polygonal paths.
        We add a vertex to the graph~$H$ for each edge in such a path that does not intersect an $F$-nerve with its relative interior, and does not intersect the $z_2$-buffer of $F$.
        A triangle (or degenerate boundary segment) vertex in $H$ is adjacent to all other regions (including triangles) with which it shares a segment on a Delaunay edge outside of $F$-nerves and the $z_2$-buffer of $F$.
        The outer-face vertex is adjacent in~$H$ to all vertices of~$G_{F,\spineset}$ that correspond to regions touching the outer face.
        These edges do not require a path in the cut and are hence irrelevant, so we immediately contract them and denote the resulting vertex by~$v_\infty$.

        \begin{figure}
            \centering%
            \includegraphics{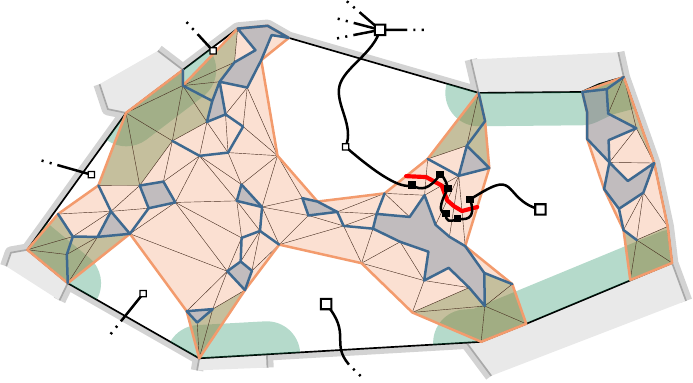}%
            \caption{
                A path in the graph $H$ that connects the two relevant vertices (large, white nodes) by a path over auxiliary vertices (small, white nodes) and triangles of the Delaunay triangulation (small, black nodes).
                The resulting planar straight-line graph $C$ is shown in red.
            }
            \label{fig:spine_cutting_ap}
        \end{figure}

        Our goal is now to find a tree $T_H$ in $H$ that spans the relevant vertices.
        Initially,~${T_H=\{v_\infty\}}$.
        To this end, we repeatedly look for a shortest path in $H$ from $v_\infty$ to a relevant vertex of~$G_{F,\spineset}$ not yet reached.
        We add this path to~$T_H$, contract the vertices of the path in $H$, and continue.
        If no such path exists, we terminate early and report that $\opt\geq \frac{1}{100}\Lambda$.

        Unless the search stops early, we turn $T_H$ into a plane, straight-line forest $C$ as follows.
        For every crossed triangulation edge, we place a vertex of $C$ at its midpoint, and inside each traversed triangle we connect the corresponding midpoints in $C$ by straight-line segments.
        If all three sides of a triangle are used, we retain only two of the three possible connections.
        Finally, we delete all connected components of $C$ that are fully located in $F$-spines that do not touch junctions.
        We then report the existence of a good cut and output~$C$.

        To prove that this approach is correct, we first show that early termination of the tree construction does indeed imply the claimed lower bound.

        \begin{claim}
            If the construction of $T_H$ stops early, then it holds that~$\opt\geq\frac{1}{100}\Lambda$.
        \end{claim}
        \begin{claimproof}
            By~\Cref{lem:the-hard-obstruction}, we know that~$\opt\geq\frac{1}{100}\Lambda$ if the relevant vertices of $G_{F,\spineset}$ do not lie in a single connected component of~$G_{F,\spineset}$.
            We show that the search only stops early if the relevant vertices are disconnected, which then implies the claim.

            Assume for contradiction that all relevant vertices of $G_{F,\spineset}$ are in the same connected component, but that at some point there is no path in $H$ from the (possibly contracted) vertex containing $v_\infty$ to any of the remaining relevant vertices of $G_{F,\spineset}$.
            Let $v_r$ be some relevant vertex that has not been reached.
            Then, because $v_r$ and $v_\infty$ lie in one common connected component of~$G_{F,\spineset}$, there is a path
            from $v_\infty$ to $v_r$ in $G_{F,\spineset}$.
            In particular, consider the shortest such path.
            Let~$u$ be the last vertex on this shortest path that is
            reachable from $v_\infty$ in $H$, and let~$v$ be its successor.
            Then there is an edge between~$u$ and~$v$ in~$G_{F,\spineset}$.
            Hence, there is a path $\pi$ in~$F\setminus (\buffer{z_1}{F})$ between the regions corresponding to~$u$ and~$v$ that intersects at most one~$F$-spine~$\spine\in\spineset$, and $\pi$ does not intersect any $F$-nerve. Note that $\pi$ does not intersect any points in $S \cap F$, as each such point is part of some $F$-nerve.
            We show that there exists a corresponding $uv$-path in~$H$, which implies that~$v$ is reachable from~$v_\infty$, a contradiction.
            By induction, the existence of a path in~$H$ from $v_\infty$ to~$v_r$ follows.

            If $\pi$ does not intersect any $F$-spine, then the regions corresponding to $u$ and $v$ share a boundary segment, and are thus adjacent in~$H$.
            This could only occur if $u$ is actually $v_\infty$ and is thus prevented by the contraction step in our construction.
            We can assume for the remainder of this proof that $\pi$ intersects exactly one $F$-spine $\spine\in\spineset$ and restrict $\pi$ to the subpath that is contained in~$\spine$.

            If part of $\spine$ is degenerate, i.e., forms a polygonal path, and $\pi$ visits the interior of an edge on such a path, then there is a path in $H$ from $u$ to $v$ via the vertex representing this edge. Indeed, by \Cref{lem:max_edge_length_support}, the maximum length of this edge is $3 r$. Using that $z_1 - z_2  = 4r - r/2> 3 r$, it follows that the edge is contained in~$\Sigma \setminus \buffer{z_2}{F}$.

            It remains to consider the case in which $\pi$ does not visit such a degenerate edge and is fully contained in a spine $\spine\in\spineset$.
            Let $\Delta$ be any triangle of the Delaunay triangulation of $S$ that is intersected by $\pi$ within $\Sigma$.
            We show that $\Delta$ is a vertex of the original graph $H$, and thus $\pi$ implies the existence of a path in the current graph $H$ between the outer-face vertex $v_\infty$ and $v$, a contradiction.

            Recall that any triangle of the Delaunay triangulation of $S$ is either fully contained in~$P$, or does not intersect its interior.
            Let $P'$ be the spine from~$\supp{\Lambda^{-1}r}$ that contains the $F$-spine~$\Sigma$.
            As $\pi$ is fully contained in $\Sigma$ and disjoint from the $F$-nerves, it must be that $\Delta$ is contained in~$P'$, and does not intersect the interior of any $F$-nerve.
            What remains is to certify that $\Delta$ is contained in~$F$, and more specifically in~$\Sigma \setminus \buffer{z_2}{F}$.
            Because $\pi$ has distance at least $z_1$ from every junction, there must be a point of $\Delta$ outside of the $z_1$-buffer.
            By \Cref{lem:max_edge_length_support}, the diameter of $\Delta$ is at most $3 r$. Using that $z_1 - z_2 = 4r - r/2 > 3 r$, it follows that~$\Delta$ is contained in~${\Sigma \setminus \buffer{z_2}{F}}$, and hence in $F$. We conclude that $\Delta$ is a vertex of $H$.
            Since consecutive traversed triangles are adjacent in~$H$, it follows by induction that we find a path from~$v_\infty$ to $v_r$.

            Hence, if there is a path between the outer-face vertex and another relevant vertex in $G_{F,\spineset}$, then the search will always find a path in $H$ that connects them.
        \end{claimproof}

        It remains to argue that $C$ is a good cut.

        We first show that $T_H$ is indeed a tree, which implies that $C$ is acyclic and plane.
        Expanding the contracted vertices in reverse order of selection, each selected path meets the previously chosen paths only at its initial vertex, so their union is a tree~$T_H$ in~$H$.
        As $T_H$ is a tree, no cycles exist in $C$ that involve more than one triangle.
        Within each triangle, at most two edges exist and do not cross.
        By construction, $C$ is thus contained in the $F$-spines in $\spineset$ and avoids every $F$-nerve and the $z_2$-buffer.
        We conclude that $C$ has properties~\ref{itm:in_spines} and~\ref{itm:not_intersect_nerve}.

        For property~\ref{itm:tree}, we need to show that the union of $C$ and all vertex regions has a connected component containing all relevant vertices and all edges of $C$.
        Consider $C$ prior to the deletion step involving connected components from $C$ that are contained in $F$-spines which do not touch junctions.
        Clearly, property~\ref{itm:tree} holds at this point, as $T_H$ is a tree containing the outer-face vertex and all relevant vertices of $G_{F,\spineset}$.
        We show that the deletion step preserves the property.
        Let $C'\subseteq C$ be a connected component that is deleted.
        Then $C'$ is contained in an $F$-spine~$\spine\in\spineset$ which does not touch any junction.
        Hence, the union of all regions touching $\spine$, potentially including the outer face, is connected; recall that $\spine$ is weakly simple.
        In particular, this means that we can walk along the boundary $\partial \spine$, with a sufficiently small offset into its complement $\mathbb{R}^2\setminus \spine$ to connect any two incidences of $C'$ with the boundary $\partial \spine$ through the enclosing polygonal regions, without crossing a spine.
        These connections pass only through regions represented by vertices, so deleting all such components from~$C$ preserves property~\ref{itm:tree}.

        We show that for every $F$-spine $\Sigma\in\spineset$ that intersects at least one junction, every connected component of $\Sigma\setminus C$ also intersects at least one junction~(property~\ref{itm:minimal}).
        Assume for contradiction there is such a spine $\Sigma$ that has a connected component~$\overline \Sigma$ of~$\Sigma\setminus C$ which does not intersect any junction.
        Note that by definition, all leaves of~$C\cap \Sigma$ are placed on the boundary~$\partial\Sigma$.
        We first show that $\overline \Sigma$ cannot be separated from all junctions by a single path in $C$.
        If~$\overline \Sigma$ does not touch the boundary of $\Sigma$, then $C$ must contain a cycle, a contradiction.
        Assume therefore that $C$ contains a path that encloses~$\overline \Sigma$ together with a single piece of the boundary~$\partial \Sigma$.
        Then this path starts and ends at regions contained in $G_{F,\spineset}$ by definition.
        Since~$\overline \Sigma$ does not touch a junction, the path must start and end at the same region; it thus induces a cycle in~$T_H$, a contradiction.
        In the remaining case, traversing the boundary of $\overline \Sigma$ yields a cyclic walk that alternates between $C$ and~$\partial \Sigma$.
        Because leaves of $C\cap\spine$ exist only on $\partial \spine$, this walk does not backtrack along an edge of $C$ before reaching $\partial \spine$.
        Replacing the boundary portions by the respective region vertex touching $\overline\spine$, and the cut portions with the respective walks in $T_H$, we obtain a non-empty cyclic walk in $T_H$.
        Within each triangle, this walk enters and leaves through distinct Delaunay edges, and at each region vertex, it leaves and re-enters $\spine$ through distinct incidences of triangles (or degenerate edges) with that region, which are contained in $T_H$.
        The walk hence does not include any immediate backtracking and gives rise to a cycle in $T_H$, a contradiction.
        We conclude that~$\overline \spine$ touches at least one junction~(property~\ref{itm:minimal}).

        Finally, by a similar argument, we can conclude that $C$ does not sever all connections between two junctions $J$ and $J'$ adjacent to~$F$ if these were connected by a path in~$\spineset\cup\junc(F)$ before.
        Suppose for contradiction that~$C$ separates~$J$ and~$J'$, and let~$C'$ refer to an inclusion-minimal subset of edges and isolated points in~$C$ that separates~$J$ and~$J'$.
        Hence, every point on $C'$ touches the component of $(\spineset\cup\junc(F))\setminus C'$ that contains $J$, as well as some components that do not.
        Then $C'$ can be associated with a non-empty subgraph of $T_H$; we show that this subgraph contains a cycle.
        In particular, we associate with $C'$ the edges of $T_H$ corresponding to proper crossings of $C'$ with the boundaries (that is, changing from or to the component containing $J$) of triangles and the regions from $G_{F,\spineset}$.
        To account for degenerate spine edges, isolated points of $C'$ on such edges are associated with both incident region edges.
        Walking around the boundary of any region corresponding to a vertex of $H$ will yield an even number of (or zero) crossings with $C'$, each either going from or to the component containing $J$.
        Hence, the vertices associated with regions that intersect~$C'$ are either isolated or have degree at least two in the subgraph of $T_H$, which thus contains a cycle.
        As $T_H$ is acyclic, we conclude that at least one path through $F$-spines remains between the junctions (property~\ref{itm:path_junctions}).
        We conclude that~$C$ is a good cut.

        There are $O(n)$ triangulation faces, and at most $O(n)$ shortest-path computations are performed in an auxiliary graph of linear size. Using a straightforward implementation, the total running time is $O(n^2\log n)$.
    \end{proof}

\end{document}